\documentclass[12pt, draftclsnofoot, onecolumn]{IEEEtran}
\usepackage{bbm}
\usepackage{caption}
\usepackage{subfigure}
\usepackage{graphicx}
\usepackage{latexsym}
\usepackage{diagbox}
\usepackage{changepage}
\usepackage[fleqn]{amsmath}
\usepackage{amsmath}
\usepackage{amsfonts}
\usepackage{indentfirst}
\usepackage{CJK}
\usepackage{indentfirst}
\usepackage[varg]{txfonts}
\usepackage{stfloats}
\usepackage{multirow}%
\usepackage{booktabs}
\usepackage{color,soul}
\usepackage{epstopdf}
\usepackage{float}
\usepackage{bm}
\usepackage{cite}
\usepackage{makecell}
\usepackage{array}
\usepackage{bm}
\usepackage{balance}
\usepackage{mathtools}
\usepackage{geometry}
\usepackage[linesnumbered,ruled,commentsnumbered,longend]{algorithm2e}
\makeatletter

\newcommand{\Rmnum}[1]{\expandafter\@slowromancap\romannumeral #1@}
\renewcommand{\figurename}{Fig.}
\makeatother
\makeatletter
\newtheorem{theorem}{Theorem}
\newtheorem{lemma}{Lemma}

\newenvironment{proof}[1][Proof]{\begin{trivlist}
		\item[\hskip \labelsep {\itshape #1}]}{\end{trivlist}}

\newcommand{\qed}{\nobreak \ifvmode \relax \else
	\ifdim\lastskip<1.5em \hskip-\lastskip
	\hskip1.5em plus0em minus0.5em \fi \nobreak
	\vrule height0.75em width0.5em depth0.25em\fi}

\begin{document}

\title{Robust Design for Intelligent Reflecting Surface Assisted MIMO-OFDMA Terahertz Communications}
\author{Wanming Hao,~\IEEEmembership{Member,~IEEE,} Gangcan Sun, Ming Zeng,~\IEEEmembership{Member,~IEEE,} Zhengyu Zhu,~\IEEEmembership{Member,~IEEE,} Zheng Chu,~\IEEEmembership{Member,~IEEE,}  Octavia A. Dobre,~\IEEEmembership{Fellow,~IEEE,}  Pei Xiao,~\IEEEmembership{Senior Member,~IEEE,}


	\thanks{W. Hao, G. Sun, and Z. Zhu are with the School of Information Engineering, Zhengzhou University, Zhengzhou 450001, China. (E-mail: \{iewmhao, iegcsun, iezyzhu\}@zzu.edu.cn)}
	\thanks{M. Zeng and O. A. Dobre are with the Faculty of Engineering and Applied Science, Memorial University, St. Johns, NL A1B 3X9, Canada. (E-mail: \{mzeng, odobre\}@mun.ca)}
	\thanks{Z. Chu and P. Xiao are with the 5G Innovation Center, Institute of Communication Systems, University of Surrey, Guildford GU2 7XH, U.K. (Email: \{zheng.chu, p.xiao\}@surrey.ac.uk)}
}

%



\maketitle
\begin{abstract}
Recently, terahertz (THz) communication has drawn considerable attention as one of the promising technologies for the future wireless communications owning to its ultra-wide bandwidth. Nonetheless, one major obstacle that prevents the actual deployment of THz lies in its inherent huge attenuation.  Intelligent reflecting surface (IRS) and multiple-input multiple-output (MIMO) represent two effective solutions for compensating the large pathloss in THz systems.  In this paper, we consider an IRS-aided multi-user THz MIMO system with orthogonal frequency division multiple access, where the sparse radio frequency chain antenna structure is adopted for reducing the power consumption. The objective is to maximize the weighted sum rate via jointly optimizing the hybrid analog/digital beamforming at the base station and reflection matrix at the IRS. {Since the analog beamforming and reflection matrix need to cater all users and  subcarriers,  it is  difficult to directly solve the formulated problem, and thus, an alternatively iterative optimization algorithm is proposed.} Specifically,  the analog beamforming is designed by solving a MIMO capacity maximization problem, while the digital beamforming and reflection matrix optimization are both tackled using semidefinite relaxation technique. Considering that obtaining perfect channel state information (CSI) is a challenging task in IRS-based systems, we further explore the case with the imperfect CSI for the channels from the IRS to users. Under this setup, we propose a robust beamforming and reflection matrix design scheme for the originally formulated non-convex optimization problem. Finally, simulation results are presented to demonstrate the effectiveness of the proposed algorithms.
\end{abstract}


%
\IEEEpeerreviewmaketitle

\section{Introduction}
With the rapid proliferation of various novel applications, such as virtual reality, augmented reality, and telemedicine, the data rate demands in future wireless communication are expected to grow explosively~\cite{Cisco}. As such, the sub-6 Gigahertz (GHz) and millimeter-wave (mmWave) may not be able to support these bandwidth-hungry applications. That being said, terahertz (THz) communication (0.1-10 THz) has been regarded as a promising technology to deal with the above problem due to its ultra-wide bandwidth~\cite{Sarieddeen_JSAC_2019,Hao_IEEESJ_2020}. However, there are two major shortcomings for THz communications, namely severe signal attenuation and poor diffraction\cite{Priebe_IEEETWC_2013}.

Multiple-input multiple-output (MIMO) has been recognized as an effective technology to enhance the THz signal strength owing to the high beamforming gain. Indeed, it has been shown that the signal strength grows linearly with the number of antennas at the base station (BS)~\cite{Lu_2014_JSTSP}. Meanwhile, the small wavelength in THz makes it easy to pack more antennas together, and form a massive MIMO array. This way, the problem of severe signal antenuation of THz can be substantially relieved. Nonetheless, the property of poor diffraction still makes THz vulnerable to blocking obstacles that break the line-of-sight (LoS) links. To address this problem, intelligent reflect surface (IRS) can be deployed to create additional links~\cite{Zeng1,Zeng2}, and thus, enhance the THz systems.  Being equipped with  a large number of reconfigurable passive elements~\cite{Wu_IEEETWC_2019,Chu_IEEEWCL_2020}, the IRS can reflect the incident signals to any direction via adjusting the phase shifts. As a result, when there is no direct link between the transmitter and receiver, communication can still be realized via building a reflective link with the help of the IRS as shown in Fig.~\ref{Systemfigure1}. Therefore, incorporating MIMO and IRS into the THz communication can effectively enhance the signal reception and reduce the probability of signal blockage.  

In this paper, we study a multi-user IRS-aided
THz MIMO system, where the BS employs sparse RF chain structure for lowering the circuit power consumption~\cite{Gao_IEEEJSAC_2016}. Meanwhile, considering that the wideband THz signals may suffer from frequency selective fading, orthogonal frequency division multiple (OFDM) is also adopted. Based on this system model, we design the hybrid analog/digital beamforming at the BS and the reflection matrix at the IRS for maximizing the weighted sum rate under perfect and imperfect channel state information (CSI). 

\subsection{Related Works}
The MIMO THz communication has become a research hotspot in recent years. Considering the large signal attenuation, Lin~{\it{et al.}} study the indoor short range MIMO THz communications~\cite{Lin_IEEETCOM_2015,Lin_IEEETWC_2016}. The authors propose a hybrid analog/digital beamforming to maximize the energy efficiency of the system. Busari~{\it{et al.}} consider three hybrid beamforming array structures, namely fully connection, subconnection and overlapped subarray~\cite{Busari_IEEETVT_2019}. Then, a single-path THz channel model is used to investigate the performance of the system under different array structures. Additionally, due to the ultra-wide bandwidth, frequency selective hybrid beamforming design in THz system is necessary. For example, Tan and Dai  first analyze the array gain loss in the wideband THz system and then propose a time delay network to obtain the near-optimal array gain\cite{Tan_IEEEGB_2019}. However, the complexity of the considered system is prohibitively high. Yuan~{\it{et al.}} build a 3-D wideband THz channel model and propose a two-stage hybrid analog/digital beamforming for maximizing the capacity of the system~\cite{Hang_IEEETCOM_2019}. After that, the imperfect CSI is also considered and a robust beamforming design scheme is developed.

In parallel, IRS has attracted  great attention in the past two years owning to its ability to enable cost-effective and energy-efficient communications. Wu and Zhang provide a basic IRS communication system model in~\cite{Wu_IEEETWC_2019}, based upon which the joint active beamforming at the BS and passive beamforming at the IRS is designed to minimize the system power consumption. 
In addition,  Ning~{\it{et al.}}  propose to apply THz to IRS~\cite{Ning_2019arXiv}, and consider the beam training and hybrid analog/digital beamforming. They propose two effective hierarchical codebooks and beamforming design schemes to obtain the near-optimal performance.  To study the performance of IRS in frequency-selective fading channels, Zhang~{\it{et al.}} consider a MIMO-OFDM system~\cite{Zhang_IEEEJSAC_2020}, where only one common set of IRS reflective matrix is designed for all subcarriers. Based on this, a new alternative optimization algorithm is proposed. Yang~{\it{et al.}} investigate the channel estimation and beamforming design problem in the IRS-based OFDM system~\cite{Yang_IEEETCOM_2020}, and propose a practical transmission protocol as well as channel estimation scheme. On this basis, a strategy of jointly optimizing power allocation and the reflection matrix is developed for maximizing the achievable rate.   

Although the THz and IRS techniques have been investigated in the literature, e.g., in \cite{Lin_IEEETCOM_2015,Lin_IEEETWC_2016,Busari_IEEETVT_2019,Tan_IEEEGB_2019,Hang_IEEETCOM_2019,Wu_IEEETWC_2019,Dai_arXiv_2020,Ning_2019arXiv,Zhang_IEEEJSAC_2020,Yang_IEEETCOM_2020}, most of them do not consider the hybrid beamforming at the BS for IRS communication~\cite{Wu_IEEETWC_2019,Dai_arXiv_2020,Ning_2019arXiv,Zhang_IEEEJSAC_2020,Yang_IEEETCOM_2020}. In fact, in a THz-based IRS communication system, the BS should employ a sparse RF antenna structure for reducing the power consumption and  the multiple subcarriers transmission technology should be adopted for overcoming the frequency selection channel fading. In this case, how to design the hybrid analog/digital beamforming at the BS and reflection matrix at the IRS catering to all subchannels will be challenging. In addition, how to obtain the perfect CSI remains a non-trivial task for IRS-based reflection links. For the direct link from the BS to users, the CSI can be readily estimated by conventional channel estimation methods. For the indirect link from the BS to the IRS, the CSI is also relatively easy to obtain since the locations of IRS and BS are fixed. However, the accurate CSIs of reflection links from the IRS to users are usually difficult to obtain due to the mobility of users. However, \cite{Wu_IEEETWC_2019,Dai_arXiv_2020,Ning_2019arXiv,Zhang_IEEEJSAC_2020,Yang_IEEETCOM_2020} all assume perfect CSI. Although Zhou~{\it{et al.}} investigate the robust beamforming design in an IRS system~\cite{Zhou_IEEEWCL_2020}, the conventional multiple antenna structure and single carrier scenario are considered. 
\subsection{Main Contributions}
To the best of our knowledge, this is the first work to consider the hybrid analog/digital beamforming in the IRS-aided THz MIMO-OFDMA under imperfect CSI, and the main contributions of this paper include:
\begin{itemize}
\item We  construct an IRS-aided THz MIMO-OFDMA communication system, where the BS employs sparse RF chain structure for reducing the circuit power consumption. First, we investigate the joint optimization of the hybrid beamforming at the BS and reflection matrix at the IRS for maximizing the weighted sum rate under perfect CSI. 
\item To solve the formulated non-trivial problem, we first initialize the reflection matrix. Since all subcarriers share one analog beamforming matrix, we ignore the multi-user interference and  obtain the analog beamforming  by
solving the corresponding  MIMO capacity optimization problem. We subsequently reformulate a multi-user weighted sum rate maximization problem to optimize the digital beamforming. With the help of successive convex approximation (SCA) and semidefinite relaxation (SDR) techniques, we propose an iterative algorithm to solve the digital beamforming that mitigates the multi-user interference. Next, we formulate the reflection matrix optimization problem under given hybrid analog/digital beamforming, and an iterative algorithm is proposed to solve it. The above procedure is repeated until convergence.   
\item Next, we assume that the perfect CSIs of reflection links cannot be obtained, and there exists bounded estimation error. We apply the same method to solve the analog beamforming. For the digital beamforming and reflection matrix, we develop a robust optimization scheme for the  weighted sum rate optimization problem relying on the $\mathcal{S}$-Procedure and the convex approximation techniques. Finally, our simulation results demonstrate the effectiveness of the proposed algorithms.
\end{itemize}

We organize the rest of this paper as follows. The system model and weighted sum rate optimization problem are introduced in~Section II. An alternatively iterative optimization algorithm is designed in Section~III. The imperfect CSIs from the IRS to users are considered and the corresponding optimization algorithm is developed in Section IV. Simulation results are presented in Section~VI. Finally, conclusions are drawn in Section~VII.

\textit{Notations}: We use the following notations throughout this paper: 
$(\cdot)^T$ and $(\cdot)^H$ denote the transpose and Hermitian transpose, respectively, $\|\cdot\|$ is the Frobenius norm, ${\mathbb{C}}^{x\times y}$ means the space of $x\times y$ complex matrix, {Re($\cdot$)} and {Tr($\cdot$)} denote real number operation and trace operation, respectively, and Diag($a_1,\ldots,a_n$) is a diagonal matrix.  $\angle(\cdot)$ represents the phase of a complex number.

\begin{figure}[t]
	\begin{center}
		\includegraphics[width=7.5cm,height=3cm]{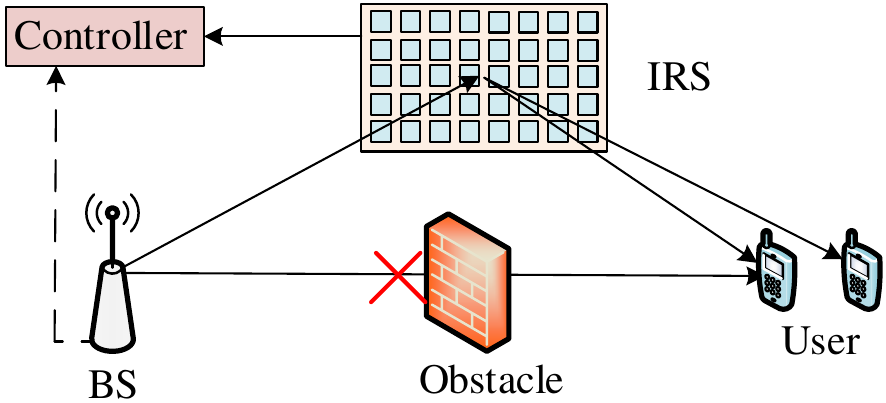}
		\caption{The IRS system model.}
		\label{Systemfigure1}
	\end{center}
\end{figure}  
\begin{figure}[t]
	\begin{center}
		\includegraphics[width=7.5cm,height=3cm]{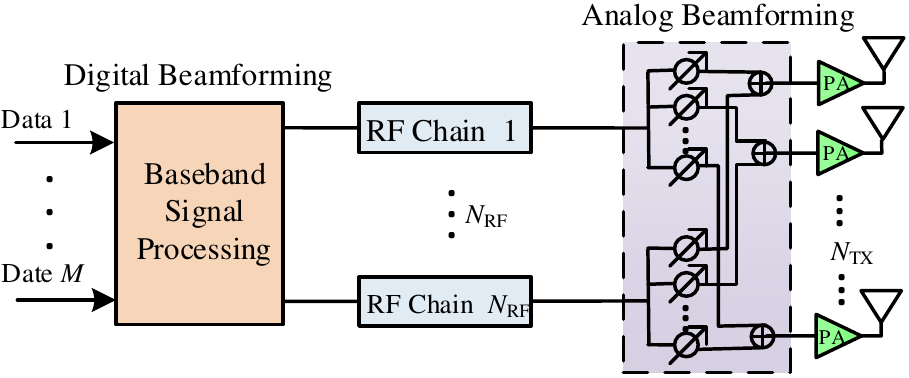}
		\caption{The sparse RF chain structure at the BS.}
		\label{Systemfigure3}
	\end{center}
\end{figure}

\section{System Model and Problem Formulation}
In this section, we first describe the IRS-aided THz MIMO-OFDMA system model and antenna structure. Next, we present the THz channel transmission model and corresponding parameters. Finally, we formulate the weighted sum rate maximization problem. 
\subsection{System Model}
We consider an  IRS-aided THz multi-user MIMO system with OFDMA as shown in Fig.~\ref{Systemfigure1}, where the BS is equipped with $N_{\rm{TX}}$ antennas and $N_{\rm{RF}}\;(N_{\rm{RF}}\leq N_{\rm{TX}})$ RF chains. The diagram of the sparse RF chain at the BS is illustrated in Fig.~\ref{Systemfigure3}. We assume that there are no direct links between BS and users due to the occlusion of walls or other obstacles, and the users can only receive the reflected signals from IRS.  Let $N_{\rm{IRS}}$, $M$ and $K$ denote the number of IRS elements, users and subcarriers, respectively.

The received signal on the $k$th subcarrier at the $m$th user can be expressed as
\begin{eqnarray}\label{eq1}
y_m[k]={\bf{G}}_m[k]{\bf{F}}{\bf{v}}_m[k]x_m[k]+\sum_{j\neq m}^{M}{\bf{G}}_m[k]{\bf{F}}{\bf{v}}_j[k]x_j[k]+n_m[k],
\end{eqnarray}
where ${\bf{G}}_m[k]=G_tG_r\eta_k{\widehat{\bf{g}}}_m[k]{\bf\Phi}{\widehat{\bf{H}}}[k]$, with $G_t$ and $G_r$ as the transmit and receive antenna gains, respectively, and $\eta_k$ as the pathloss compensation factor~\cite{Ning_2019arXiv}. ${\widehat{\bf{g}}}_m[k]\in {\mathbb{C}}^{1\times N_{\rm{IRS}}}$ denotes the channel vector from IRS to the $m$th user on the $k$th subcarrier, ${\bf\Phi}\in {\mathbb{C}}^{N_{\rm{IRS}}\times N_{\rm{IRS}}}$ is the reflection coefficient matrix with  ${\bf\Phi}={\rm{diag}}\{\phi_1,...,\phi_{N_{\rm{IRS}}}\}$, ${\widehat{\bf{H}}}[k]\in {\mathbb{C}}^{N_{\rm{IRS}}\times N_{\rm{TX}}}$ represents the channel matrix from BS to IRS on the $k$th subcarrier, ${\bf{F}}\in {\mathbb{C}}^{N_{\rm{TX}}\times N_{\rm{RF}}}$ is the analog beamforming matrix with ${\bf{F}}=[{\bf{f}}_1,...,{\bf{f}}_{N_{\rm{RF}}}]$, ${\bf{v}}_m[k]\in {\mathbb{C}}^{N_{\rm{RF}}\times 1}$  and $x_m[k]$ denote the digital beamforming and transmit signal for the $m$th user on the $k$th subcarrier,  respectively, $n_m[k]$ is the independent and identically distributed (i.i.d.) additive white Gaussian noise (AWGN)  with zero-mean and variance  $N_0$. In (\ref{eq1}), the first term is the designed signal, while the second term is the multi-user interference that must be mitigated by designing proper digital beamforming and reflection matrix.

Next, we present  the THz channel model. Let $f_c$ and $B$, respectively, represent the central frequency and bandwidth. Then, the frequency band of the $k$th subcarrier can be expressed as $f_k=f_c+\frac{B}{K}(k-1-\frac{K-1}{2}), k=1,2,...,K$. Although there are a few scattering components in THz communication, their power are much lower (more than 20 dB) than that of LoS component~\cite{Priebe_JCN_2013}, and thus, we only consider the LoS component and ignore the other scattering components. Accordingly, the channel matrix ${\widehat{\bf{H}}}[k]$ can be expressed as
\begin{eqnarray}
{\widehat{\bf{H}}}[k]=q(f_k,d){{\bf{H}}}[k],
\end{eqnarray} 
where $q(f_k,d)$ is the complex path gain satisfying
\begin{eqnarray}
	q(f_k,d)=\frac{c}{4\pi fd}e^{-\frac{1}{2}\tau(f_k)d},
\end{eqnarray}
where $c$ stands for the speed of light, $\tau(f_k)$ represents the medium absorption factor and $d$ is the distance from the BS to IRS~\cite{Cui_2019arXiv}. ${{\bf{H}}}[k]$ can be expressed as
\begin{eqnarray}
{{\bf{H}}}[k]={\bf{a}}_r(\theta_{k}){\bf{a}}_t^H(\varphi_{k}),
\end{eqnarray}
with ${\bf{a}}_t(\theta_{k})$ and ${\bf{a}}_r(\varphi_{k})$, respectively, as the antenna array response vector of the transmitter and receiver, namely
\begin{subequations}
\begin{align}
\;\;\;\;\;\;\;\;\;\;\;\;\;\;\;\;\;\;\;\;\;\;\;\;\;\;\;\;\;\;&{\bf{a}}_t(\theta_{k})=\frac{1}{\sqrt{N_{\rm{TX}}}}\left[1,e^{j\pi\theta_{k}},e^{j2\pi\theta_{k}},\cdots,e^{j(N_{\rm{TX}}-1)\pi\theta_{k}}\right]^T,\\
&{\bf{a}}_r(\varphi_{k})=\frac{1}{\sqrt{N_{\rm{IRS}}}}\left[1,e^{j\pi\varphi_{k}},e^{j2\pi\varphi_{k}},\cdots,e^{j(N_{\rm{IRS}}-1)\pi\varphi_{k}}\right]^T.
\end{align}
\end{subequations}
Here, $\theta_{k}=2d_0f_k\sin(\phi_t)/c$ and $\varphi_{k}=2d_0f_k\sin(\phi_r)/c$, $d_0$ denotes the antenna distance, and $\phi_t/\phi_r\in[-\pi/2,\pi/2]$ are, respectively, angle of departure (AoD) and angle of arrival (AoA). Similarly, ${\bf{g}}_m[k]$ can be expressed~as
\begin{eqnarray}
	{\widehat{\bf{g}}}_m[k]=q(f_k,d_m){{\bf{g}}}_m[k],
\end{eqnarray}
where ${{\bf{g}}}_m[k]=\frac{1}{\sqrt{N_{\rm{IRS}}}}\left[1,e^{j\pi\varphi_{k,m}},e^{j2\pi\varphi_{k,m}},\cdots,e^{j(N_{\rm{IRS}}-1)\pi\varphi_{k,m}}\right]$, and $q(f_k,d_m)$ is defined as 
\begin{eqnarray}
q(f_k,d_m)=\frac{c}{4\pi f_kd_m}e^{-\frac{1}{2}\tau(f_k)d_m},
\end{eqnarray}
with  $d_m$ as the distance from the IRS to the $m$th user. The BS-IRS-user $m$ link channel can such be expressed~as 
\begin{eqnarray}
	{\bf{G}}_m[k]=u_m[k]{{\bf{g}}}_m[k]{\bf\Phi}{{\bf{H}}}[k],
\end{eqnarray}
where $u_m[k]=G_tG_r\eta_kq(f,d_m)q(f_k,d_m)$. According to~\cite{Cui_2019arXiv}, the cascaded path loss of the BS-IRS-user link should satisfy 
\begin{eqnarray}
	\eta_k q(f_k,d)q(f_k,d_m)=\frac{\chi c}{8\sqrt{\pi^3}f_kdd_m}e^{-\frac{1}{2}\tau(f_k)(d+d_m)},
\end{eqnarray}
where $\chi$ is the IRS element gain. 

Finally, we rewrite~(\ref{eq1}) as
\begin{equation}\label{eq11}
\begin{aligned}
	y_m[k]=u_m[k]{{\bf{g}}}_m[k]{\bf\Phi}{{\bf{H}}}[k]{\bf{F}}{\bf{v}}_m[k]x_m[k]
	+\sum\nolimits_{j\neq m}^{M}u_m[k]{{\bf{g}}}_m[k]{\bf\Phi}{{\bf{H}}}[k]{\bf{F}}{\bf{v}}_j[k]x_j[k]+n_m[k].
	\end{aligned}
\end{equation}
\subsection{Problem Formulation}
By employing (\ref{eq11}), the achievable rate of the $m$th user on the $k$th subcarrier can be expressed~as
\begin{eqnarray}
	R_m[k]=\frac{B}{K}\log\left(1\!+\!\frac{\left|u_m[k]{{\bf{g}}}_m[k]{\bf\Phi}{{\bf{H}}}[k]{\bf{F}}{\bf{v}}_m[k]\right|^2}{\sum\nolimits_{j\neq m}^{M}\left|u_m[k]{{\bf{g}}}_m[k]{\bf\Phi}{{\bf{H}}}[k]{\bf{F}}{\bf{v}}_j[k]\right|^2\!+\!BN_0/K}\right),
\end{eqnarray}
and thus, the achievable sum rate for the $m$th user can be written~as 
\begin{eqnarray}
	R_m=\sum\nolimits_{k=1}^{K}R_m[k].
\end{eqnarray}

Next, we formulate the weighted sum rate maximization problem as follows
\setlength{\mathindent}{0cm}
\begin{subequations}\label{OptA}
	\begin{align}
	\;\;\;\;\;\;\;\;\;\;\;\;\;\;\;\;\;\;\;\;\;\;\;\;\;\;\;\;\;\;\;\;&\underset{\left\{{\bf\Phi},{\bf{F}},{\bf{v}}_m[k]\right\}}{\rm{max}}\;\sum\nolimits_{m=1}^{M}\alpha_mR_m \label{OptA0}\\
	{\rm{s.t.}}&\;|\phi_i|=1,i\in \{1,\cdots,N_{\rm{IRS}}\}, \label{OptA1}\\
	&\;\sum\nolimits_{M=1}^M\sum\nolimits_{k=1}^K||{\bf{F}}{\bf{v}}_m[k]||^2\leq P_{\rm{max}},\label{OptA2}\\
	&\;{\bf{F}}(i,j)=1/\sqrt{N_{\rm{TX}}},i\in \{1,\cdots,N_{\rm{TX}}\}, j\in\{1,\cdots,N_{\rm{RF}}\}, \label{OptA3}
	\end{align}
\end{subequations} 
where $\alpha_m$ denotes the weight of the $m$th user, (\ref{OptA1}) is the uni-modular constraint for each reflection coefficient $\phi_i$, (\ref{OptA2}) is the sum transmit power constraint, and  (\ref{OptA3}) is the amplitude constraint of analog beamforming. The objective of  (\ref{OptA}) is to jointly optimize the reflection matrix ${\bf\Phi}$ and hybrid analog/digital beamforming ${\bf{F}}$ and ${\bf{v}}_m[k]$ for maximizing the weighed sum rate of the system. Obviously, (\ref{OptA}) is a non-convex optimization problem due to the non-convex objective function (\ref{OptA0}) and constrains (\ref{OptA1}), (\ref{OptA3}). Finding the optimal solution is a challenging task, and we propose an effective alternatively iterative optimization algorithm to deal with it. 
\section{Solution of The Weighted Sum Rate Optimization Problem}
In this section, we propose an alternatively iterative optimization algorithm. First, we consider the hybrid analog/digital beamforming of ${\bf{F}}$, ${\bf{v}}_m[k]$ under given reflection matrix ${\bf\Phi}$. We formulate a MIMO-OFDM sum rate maximization problem and use it to obtain the  analog beamforming matrix ${\bf{F}}$. Next, we transform the original problem into a SDP and solve the digital beamforming ${\bf{v}}_m[k]$ by the SDR technique. Lastly, we solve the reflection matrix ${\bf\Phi}$ according to the obtained hybrid beamforming ${\bf{F}}$ and ${\bf{v}}_m[k]$. The above procedure is repeated until convergence. 
\subsection{Optimization of ${\bf{F}}$ and ${\bf{v}}_m[k]$ under Fixed ${\bf\Phi}$}\label{A1}
Under given ${\bf\Phi}$,  the original problem can be transformed as 
\setlength{\mathindent}{0cm}
\begin{subequations}\label{OptB}
	\begin{align}
	\;\;\;\;\;\;\;\;\;\;\;\;\;\;\;\;\;\;\;\;\;\;\;\;\;\;\;\;\;\;\;\;\;\;\;\;&\underset{\left\{{\bf{F}},{\bf{v}}_m[k]\right\}}{\rm{max}}\;\sum_{m=1}^{M}\sum_{k=1}^{K}a_m\log\left(1\!+\!\frac{\left|{\hat{\bf{h}}}_m[k]{\bf{F}}{\bf{v}}_m[k]\right|^2}{\sum\nolimits_{j\neq m}^{M}\left|{\hat{\bf{h}}}_m[k]{\bf{F}}{\bf{v}}_j[k]\right|^2\!+\!\delta^2}\right) \label{OptB0}\\
	{\rm{s.t.}}	&\;\rm(\ref{OptA2}),\rm(\ref{OptA3}),
	\end{align}
\end{subequations} 
where ${\hat{\bf{h}}}_m[k]=u_m[k]{{\bf{g}}}_m[k]{\bf\Phi}{{\bf{H}}}[k]$, $a_m=B\alpha_m/K$, and $\delta^2=BN_0/K$. 

Problem (\ref{OptB}) is still difficult to solve since each common element of the analog beamforming ${\bf{F}}$ needs to cater all users and subcarriers. Furthermore, the data streams of different users may own different priority weights, which leads to more complicated analog beamforming design. For simplicity, we assume that all users have the same weights and it can be  regarded as a MIMO-OFDM system by neglecting the inter-user interference. To solve the analog beamforming ${\bf{F}}$, we first  define ${\hat{\bf{H}}}[k]=[{\hat{\bf{h}}}_1[k]^T,...,{\hat{\bf{h}}}_M[k]^T]^T$, and reformulate a MIMO-OFDM sum rate maximization problem as follows:
\setlength{\mathindent}{0.5cm}
\begin{subequations}\label{OptC}
	\begin{align}
	\;\;\;\;\;\;\;\;\;\;\;\;\;\;\;\;\;\;\;\;\;\;\;\;\;\;\;\;\;\;\;\;&\underset{\left\{{\bf{F}},{\bf{V}}[k]\right\}}{\rm{max}}\;\frac{1}{K}\sum_{k=1}^{K}\log\left({\bf{I}}\!+\!\frac{{\hat{\bf{H}}}[k]{\bf{F}}{\bf{V}}[k]{\bf{V}}[k]^H{\bf{F}}^H{{\hat{\bf{H}}}[k]^H}}{\!\delta^2}\right) \label{OptC0}\\
	{\rm{s.t.}}	&\;||{\bf{F}}{\bf{V}}[k]||^2\leq P_{\rm{max}}/K, \rm(\ref{OptA3}). \label{OptC1}
	\end{align}
\end{subequations} 

Here, we consider the transmit power constraint for each subcarrier to obtain the low bound. Note that (\ref{OptC}) is only used for  designing ${\bf{F}}$. For a given ${\bf{F}}$, the optimal digital beamforming matrix can be calculated as~\cite{Park_IEEETWC_2017}
\begin{eqnarray}\label{eq16}
	{\bf{V}}[k]=({\bf{F}}^H{\bf{F}})^{-{1}/{2}}{\bf{U}}_e[k]{\bf{\Gamma}}_e[k],
\end{eqnarray}
where ${\bf{U}}_e[k]$  is the set of right singular vector according to the $N_{\rm{RF}}$ largest singular value  of ${\hat{\bf{H}}}[k]{\bf{F}}({\bf{F}}^H{\bf{F}})^{-{1}/{2}}$, and ${\bf{\Gamma}}_e[k]$ is a diagonal matrix of the power allocated to the data streams on each subcarrier. Here, we assume that there are $N_{\rm{RF}}$ data streams. In addition, it is obvious that ${\bf{f}}_i^H{\bf{f}}_i=1$, while ${\bf{f}}_i^H{\bf{f}}_j\ll1\;(i\neq j)$ with high probability for a large $N_{\rm{TX}}$. Therefore, the analog beamforming satisfies ${\bf{F}}^H{\bf{F}}\approx {\bf{I}}$ that can always be approximated as proportional to the identity matrix, namely ${\bf{F}}^H{\bf{F}} \propto {\bf{I}}$.  Moreover, in the high or moderate signal-to-noise ratio (SNR) regime, the equal power allocation scheme for all streams on each subcarrier can be adopted without significant performance degradation, namely ${\bf{\Gamma}}_e[k] \propto {\bf{I}}$~\cite{shen_IEEETSP_2019}. As a result, the digital beamforming matrix can be approximated as ${\bf{V}}[k]\approx \lambda {\bf{U}}_e[k]$, where $\lambda=\sqrt{P_{\rm{max}}/(KN_{\rm{TX}}N_{\rm{RF}})}$. Based on the above analysis, we have
\begin{equation}\label{eq18}
\begin{aligned}
	\;\;\;\;\;\;\;\;\;\;\;\;\;\;\;\;\;\;\;\;\;\;\;\;\;\;\;\;\;\;\;\;\;\;\;\;\;\;\;\;&\frac{1}{K}\sum_{k=1}^{K}\log\left({\bf{I}}\!+\!\frac{{\hat{\bf{H}}}[k]{\bf{F}}{\bf{V}}[k]{\bf{V}}[k]^H{\bf{F}}^H{{\hat{\bf{H}}}[k]^H}}{\!\delta^2}\right)\\
	=&\frac{1}{K}\sum_{k=1}^{K}\log\left({\bf{I}}\!+\!\frac{\lambda}{\!\delta^2}{\hat{\bf{H}}}[k]{\bf{F}}{\bf{U}}_e[k]{\bf{U}}_e[k]^H{\bf{F}}^H{{\hat{\bf{H}}}[k]^H}\right)\\
	=&\frac{1}{K}\sum_{k=1}^{K}\log\left({\bf{I}}\!+\!\frac{\lambda}{\!\delta^2}{\bf{F}}^H{\hat{\bf{H}}}[k]^H{{{\hat{\bf{H}}}[k]}\bf{F}}\right).
	\end{aligned}
\end{equation}

Finally, we can obtain the upper bound of (\ref{eq18}) using the Jensen's inequality as
\begin{eqnarray}\label{eq19}
\frac{1}{K}\sum_{k=1}^{K}\log\left({\bf{I}}\!+\!\frac{\lambda}{\!\delta^2}{\bf{F}}^H{\hat{\bf{H}}}[k]^H{{{\hat{\bf{H}}}[k]}\bf{F}}\right)\leq \log\left({\bf{I}}\!+\!\frac{\lambda}{\!\delta^2}{\bf{F}}^H{{{{\bf{\Sigma}}}}\bf{F}}\right),
\end{eqnarray}
where ${{{\bf{\Sigma}}}}=\frac{1}{K}\sum_{k=1}^{K}\left({\hat{\bf{H}}}[k]^H{\hat{\bf{H}}}[k]\right)$ and the analog beamforming matrix can be obtained by solving the following problem
\begin{eqnarray}\label{eq20}
{\bf{F}}^\star=\underset{{\bf{F}}(i,j)=1/\sqrt{N_{\rm{TX}}}}{\rm{arg\;max}}\;\;\;\;\;\log\left({\bf{I}}\!+\!\frac{\lambda}{\!\delta^2}{\bf{F}}^H{{{{\bf{\Sigma}}}}\bf{F}}\right).
\end{eqnarray}
Since ${{{\bf{\Sigma}}}}$ is a hermitian matrix, its singular value decomposition (SVD) can be written as
	${{{\bf{\Sigma}}}}={\bf{S}}{\bf{\Lambda}}{\bf{S}}^H.$
Therefore, the solution of (\ref{eq20}) can be given by ${\bf{F}}^\star(i,j)=\frac{1}{\sqrt{N_{\rm{TX}}}}e^{\angle\left({\bf{S}}_{1:N_{\rm{BF}}}(i,j)\right)}$, where ${\bf{S}}_{1:N_{\rm{BF}}}$ denotes the first $N_{\rm{BF}}$ columns of ${\bf{S}}$.

Note that although (\ref{eq16}) provides a digital beamforming solution, it is only suitable for the single user case. Therefore, after obtaining the analog beamforming matrix ${\bf{F}}$, problem (\ref{OptB}) can be transformed as follows for solving digital beamforming
\setlength{\mathindent}{0cm}
\begin{subequations}\label{OptD}
	\begin{align}
	\;\;\;\;\;\;\;\;\;\;\;\;\;\;\;\;\;\;\;\;\;\;\;\;\;\;\;\;\;\;\;\;\;\;\;\;\;&\underset{\left\{{\bf{v}}_m[k]\right\}}{\rm{max}}\;\sum_{m=1}^{M}\sum_{k=1}^{K}a_m\log\left(1\!+\!\frac{\left|{\bar{\bf{h}}}_m[k]{\bf{v}}_m[k]\right|^2}{\sum\nolimits_{j\neq m}^{M}\left|{\bar{\bf{h}}}_m[k]{\bf{v}}_j[k]\right|^2\!+\!\delta^2}\right) \label{OptD0}\\
	\;\;\;\;\;\;{\rm{s.t.}}	&\;\sum\nolimits_{M=1}^M\sum\nolimits_{k=1}^K||{\bf{F}}{\bf{v}}_m[k]||^2\leq P_{\rm{max}},
	\end{align}
\end{subequations} 
where ${\bar{\bf{h}}}_m[k]={\hat{\bf{h}}}_m[k]{\bf{F}}$. We define ${\bar{\bf{H}}}_m[k]={\bar{\bf{h}}}_m[k]^H{\bar{\bf{h}}}_m[k]$ and ${\bf{V}}_m[k]={\bf{v}}_m[k]{\bf{v}}_m[k]^H$, and by using an auxiliary variable $t_{m,k}$, we reformulate (\ref{OptD}) as
\begin{subequations}\label{OptE}
	\begin{align}
	\;\;\;\;\;\;\;\;\;\;\;\;\;\;\;\;\;\;\;\;\;\;\;\;\;\;\;\;\;\;\;\;\;\;\;\;\;\;\;&\underset{\left\{{\bf{V}}_m[k]\right\}}{\rm{max}}\;\sum_{m=1}^{M}\sum_{k=1}^{K}a_m\log\left(1\!+\!t_{m,k}\right) \label{OptE0}\\
	\;\;\;\;\;\;\;\;\;{\rm{s.t.}}	&\;t_{m,k}\leq \frac{{\rm{Tr}}({\bar{\bf{H}}}_m[k]{\bf{V}}_m[k])}{\sum\nolimits_{j\neq m}^{M}{\rm{Tr}}({\bar{\bf{H}}}_m[k]{\bf{V}}_j[k])\!+\!\delta^2},\label{OptE1}\\
	&\;\sum\nolimits_{M=1}^M\sum\nolimits_{k=1}^K{\rm{Tr}}({\bf{F}}^H{\bf{F}}{\bf{V}}_m[k])\leq P_{\rm{max}},\label{OptE2}\\
	&\;{\rm{rank}}({\bf{V}}_m[k])=1, {\bf{V}}_m[k]\succeq 0. \label{OptE3}
	\end{align}
\end{subequations} 

It is obvious that (\ref{OptE}) is a non-convex optimization problem due to (\ref{OptE1}) and  (\ref{OptE3}). To cope with (\ref{OptE1}), we introduce an auxiliary variable $b_{m,k}$ and transform it as
\begin{subequations}
\begin{align}
	\;\;\;\;\;\;\;\;\;\;\;\;\;\;\;\;\;\;\;\;\;\;\;\;\;\;\;\;\;\;\;\;\;\;\;\;\;\;\;\;\;\;\;\;\;\;\;\;\;\;\;\;\;\;\;\;\;&t_{m,k}b_{m,k}\leq {\rm{Tr}}({\bar{\bf{H}}}_m[k]{\bf{V}}_m[k]),\label{eq23A}\\
	&b_{m,k}\geq \sum\nolimits_{j\neq m}^{M}{\rm{Tr}}({\bar{\bf{H}}}_m[k]{\bf{V}}_j[k])\!+\!\delta^2. \label{eq23B}
\end{align}
\end{subequations}	

Now, we only need to deal with $t_{m,k}b_{m,k}$. According to~\cite{Song_IEEEICASSP_2016}, the upper bound of $t_{m,k}b_{m,k}$ can be obtained~as
\begin{eqnarray}
	\frac{t_{m,k}^{(n)}}{2b_{m,k}^{(n)}}b_{m,k}^2+\frac{b_{m,k}^{(n)}}{2t_{m,k}^{(n)}}t_{m,k}^2\geq t_{m,k}b_{m,k}, 
\end{eqnarray} 
where $t_{m,k}^{(n)}$ and $b_{m,k}^{(n)}$ are the values of $t_{m,k}$ and $b_{m,k}$ at the $n$th iteration, respectively. Consequently, we transform (\ref{eq23A}) into the following convex constraints
\begin{eqnarray}\label{eq25}
	\frac{t_{m,k}^{(n)}}{2b_{m,k}^{(n)}}b_{m,k}^2+\frac{b_{m,k}^{(n)}}{2t_{m,k}^{(n)}}t_{m,k}^2\leq {\rm{Tr}}({\bar{\bf{H}}}_m[k]{\bf{V}}_m[k]).
\end{eqnarray}

Finally, (\ref{OptE}) can be recast as the following SDP problem
\begin{subequations}\label{OptH}
	\begin{align}
	&\underset{\left\{t_{m,k}, b_{m,k},{\bf{V}}_m[k]\right\}}{\rm{max}}\;\sum_{m=1}^{M}\sum_{k=1}^{K}a_m\log\left(1\!+\!t_{m,k}\right) \label{OptH0}\\
\;\;\;\;\;\;\;\;\;\;\;\;\;\;\;\;\;\;\;\;\;\;\;\;\;\;\;\;\;\;\;\;\;\;\;\;\;\;\;\;\;{\rm{s.t.}}	&\;{\rm{(\ref{OptE2}), (\ref{eq23B}), (\ref{eq25})}}\label{OptH2},\\
	&\;{\rm{rank}}({\bf{V}}_m[k])=1, {\bf{V}}_m[k]\succeq 0. \label{OptH3}
	\end{align}
\end{subequations} 

Since the rank-one constraint is non-convex, we need to drop it and formulate a SDR problem that can be solved by  existing convex solvers such as the CVX toolbox. 
Summarily, to obtain the digital beamforming matrix ${\bf{V}}_m[k]$, we need to iteratively solve (\ref{OptH}). Specifically, we first initialize the auxiliary variables $b_{m,k}^{\rm{o}}$, $t_{m,k}^{\rm{o}}$ and solve (\ref{OptH}) for obtaining the optimal $b_{m,k}^\star$, $t_{m,k}^\star$ and ${\bf{V}}_m[k]^\star$. Next, $b_{m,k}^{\rm{o}}$ and  $t_{m,k}^{\rm{o}}$ are updated with the obtained $b_{m,k}^\star$ and $t_{m,k}^\star$, and then, we resolve  (\ref{OptH}). The above procedure is repeated until the results converge or the iteration number reaches its maximum value. In addition, since the SDR problem of (\ref{OptH}) is a convex optimization problem, the solutions are optimal for each iteration. Therefore, iteratively solving (\ref{OptH}) and updating variables increase or at least maintain the value of the objective function~\cite{Qi_IEEETVT_2016,Pan_IEEETWC_2018}. Given the limited transmit power, the designed iterative algorithm guarantees the value of the objective function to be a monotonically non-decreasing sequence with an upper bound, and it converges to a stationary solution that is at least a local optimal.   

For solving (\ref{OptH}), we remove the rank-one constraint ${\rm{rank}}({\bf{V}}_m[k])=1$. To explore the characteristic of the solutions, we first give the following theorem for the obtained digital beamforming  ${\bf{V}}^\star_m[k]$. 
\begin{theorem}\label{theorem1}
	For a large number of BS antenna $N_{\rm{TX}}$, the obtained digital beamforming ${\bf{V}}^\star_m[k]$ satisfies ${\rm{rank}}({\bf{V}}^\star_m[k])=1$.
\end{theorem}
\begin{proof}
	Refer to Appendix A.
\end{proof}

In fact, even for a medium number of BS antennas, such as $N_{\rm{TX}}=16$, we find that the optimal ${\bf{V}}^\star_m[k]$ always satisfies ${\rm{rank}}({\bf{V}}^\star_m[k])=1$. If ${\rm{rank}}({\bf{V}}^\star_m[k])=1$, the optimal ${\bf{v}}^\star_m[k]$ can be recovered by the eigenvalue decomposition, namely ${\bf{V}}^\star_m[k]={\bf{v}}^\star_m[k]^H{\bf{v}}^\star_m[k]$. If the obtained ${\bf{V}}^\star_m[k]$ is not a rank-one matrix, the Gaussian randomization technique is typically used to obtain a rank-one solution~\cite{Luo_IEEESPM_2010}.  Specifically, we first generate a random vector $\hat{{\bf{v}}}_m[k]$ satisfying  $\hat{{\bf{v}}}_m[k]\sim {\mathcal{CN}}({\bf{0}},{\bf{V}}^\star_m[k])$. Then, we apply $\hat{{\bf{v}}}_m[k]$ to problem~(\ref{OptH}) and check its feasibility. Here, we need to independently generate enough feasible beamforming vectors $\hat{{\bf{v}}}_m[k]$ and select the optimal one $\hat{{\bf{v}}}_m[k]^\star$ from all the random vectors. Based on this, the final digital beamforming vector can be approximated as ${{\bf{v}}}_m[k]^\star=\hat{{\bf{v}}}_m[k]^\star$.

Now, we analyze the computational complexity for solving problem (\ref{OptH}). Let $\epsilon$ be the iteration accuracy; then, the number of iterations is on the order of $\sqrt{KMN_{\rm{RF}}+5KM+1}\ln(1/\epsilon)$~\cite{Chi_IEEETSP_2014}. For problem~(\ref{OptH}), there are $KM$ linear matrix inequality (LMI) constraints of size $N_{\rm{RF}}$, $3KM+1$ LMI constraints of size $1$, and $KM$ second-order cone constraints. Therefore, the complexity of solving problem (\ref{OptH}) is on the order of 
$\sqrt{\Delta_1}\zeta\ln(1/\epsilon)(\Delta_2+\zeta\Delta_3+\zeta^2)$,
where $\Delta_1=KMN_{\rm{RF}}+5KM+1$, $\Delta_2=KMN_{\rm{RF}}^3+12KM+1$, $\Delta_3=KMN_{\rm{RF}}^2+3KM+1$, and $\zeta=KMN_{\rm{RF}}^2$.

\subsection{Optimization of $\bf{\Phi}$ under Fixed ${\bf{F}}$ and ${\bf{V}}_m[k]$}\label{B}
After obtaining the hybrid analog/digital beamforming ${\bf{F}}$ and ${\bf{v}}_m[k]$, we consider the reflection matrix of IRS here and transform (\ref{OptA}) into the following optimization problem
\setlength{\mathindent}{0cm}
\begin{subequations}\label{OptI}
	\begin{align}
	\;\;\;\;\;\;\;\;\;\;\;\;\;\;\;\;\;\;\;\;\;\;\;\;\;\;\;\;\;\;\;\;\;\;&\underset{\left\{{\bf{\Phi}}\right\}}{\rm{max}}\;\sum_{m=1}^{M}\sum_{k=1}^{K}a_m\log\left(1\!+\!\frac{\left|{{\bf{g}}}_m[k]{\bf\Phi}{\bf{z}}_m[k]\right|^2}{\sum\nolimits_{j\neq m}^{M}\left|{{\bf{g}}}_m[k]{\bf\Phi}{\bf{z}}_j[k]\right|^2\!+\!\delta^2}\right) \label{OptI0}\\
	&{\rm{s.t.}}\;\;|\phi_i|=1,i\in \{1,\cdots,N_{\rm{IRS}}\}, \label{OptI1}
	\end{align}
\end{subequations} 
where ${\bf{z}}_m[k]=u_m[k]{{\bf{H}}}[k]{\bf{F}}{\bf{v}}_m[k]$. Let ${{\bf{g}}}_m[k]{\bf\Phi}{\bf{z}}_m[k]={\bf{c}}_{m,k}{\bm{\phi}}$, where ${\bm{\phi}}=[\phi_1,...,\phi_{N_{\rm{IRS}}}]^T$ and ${\bf{c}}_m[k]={{\bf{g}}}_m[k]{\rm{diag}}({\bf{z}}_m[k])$. Thus, we can reformulate (\ref{OptI}) as
\begin{subequations}\label{OptJ}
	\begin{align}
	\;\;\;\;\;\;\;\;\;\;\;\;\;\;\;\;\;\;\;\;\;\;\;\;\;\;\;\;\;\;\;\;\;\;&\underset{\left\{{\bm{\phi}}\right\}}{\rm{max}}\;\sum_{m=1}^{M}\sum_{k=1}^{K}a_m\log\left(1\!+\!\frac{\left|{\bf{c}}_m[k]{\bm{\phi}}\right|^2}{\sum\nolimits_{j\neq m}^{M}\left|{\bf{c}}_j[k]{\bm{\phi}}\right|^2\!+\!\delta^2}\right) \label{OptJ0}\\
	&{\rm{s.t.}}\;\;|\phi_i|=1,i\in \{1,\cdots,N_{\rm{IRS}}\}. \label{OptJ1}
	\end{align}
\end{subequations} 

Similarly, we define ${\bf{C}}_m[k]={\bf{c}}_m[k]^H{\bf{c}}_m[k]$ and ${\bf{\Omega}}={\bm{\phi}}{\bm{\phi}}^H$,  and transform (\ref{OptJ}) into the following optimization problem
\begin{subequations}\label{OptM}
	\begin{align}
\;\;\;\;\;\;\;\;\;\;\;\;\;\;\;\;\;\;\;\;\;\;\;\;\;\;\;\;\;\;\;\;	\;\;&\underset{\left\{{\bm{\Omega}}\right\}}{\rm{max}}\;\sum_{m=1}^{M}\sum_{k=1}^{K}a_m\log\left(1\!+\!\frac{{\rm{Tr}}({\bf{C}}_m[k]{\bf{\Omega}})}{\sum\nolimits_{j\neq m}^{M}{\rm{Tr}}({\bf{C}}_j[k]{\bf{\Omega}})\!+\!\delta^2}\right) \label{OptM0}\\
	&{\rm{s.t.}}\;\;{\bf{\Omega}}(i,i)=1,i\in \{1,\cdots,N_{\rm{IRS}}\}, \label{OptM1}\\
	&\;\;\;\;\;\;\;{\rm{rank}}({\bf{\Omega}})=1, {\bf{\Omega}}\succeq 0.
	\end{align}
\end{subequations} 

One can observe that (\ref{OptM}) has a similar form to (\ref{OptE}), and we can adopt the same iterative method proposed in Section.~\ref{A1}  to solve it. We omit the details due to the space limitation and directly write the SDR problem as follows:
\begin{subequations}\label{OptL}
	\begin{align}
\;\;\;\;\;\;\;\;\;\;\;\;\;\;\;\;\;\;\;\;\;\;\;\;\;\;\;\;\;\;\;\;\;\;\;\;\;\;\;&\underset{\left\{t'_{m,k}, b'_{m,k},{\bf{\Omega}}\right\}}{\rm{max}}\;\sum_{m=1}^{M}\sum_{k=1}^{K}a_m\log\left(1\!+\!t'_{m,k}\right) \label{OptL0}\\
	{\rm{s.t.}}	&\;b'_{m,k}\geq \sum\nolimits_{j\neq m}^{M}{\rm{Tr}}({{\bf{C}}}_j[k]{\bf{\Omega}})\!+\!\delta^2,\\
	&\;\frac{{t'}_{m,k}^{(n)}}{2{b'}_{m,k}^{(n)}}{b'}_{m,k}^{2}+\frac{{b'}_{m,k}^{(n)}}{2{t'}_{m,k}^{(n)}}{t'}_{m,k}^{2}\leq {\rm{Tr}}({\bf{C}}_m[k]{\bf{\Omega}}),\\
	&\;{\bf{\Omega}}(i,i)=1, {\bf{\Omega}}\succeq 0. \label{OptL3}
	\end{align}
\end{subequations} 

Here, we still need to alternatively solve problem~(\ref{OptL}) to obtain the reflection matrix ${\bf{\Phi}}$. If the obtained ${\bm{\phi}}^\star$ satisfies ${\rm{rank}}({\bf{\Omega}}^\star)=1$, the optimal ${\bm{\phi}}^\star$ can be recovered by the eigenvector of  ${\bf{\Omega}}^\star={\bm{\phi}}^\star{\bm{\phi}}^{\star H}$, and the reflection coefficient matrix ${\bf{\Phi}}^\star$ can be expressed as ${\bf\Phi}^\star ={\rm{diag}}\{\phi_1^\star,...,\phi_{N_{\rm{IRS}}}^\star\}$. Otherwise, the Gaussian randomization technique can be  used to obtain a rank-one solution~\cite{Luo_IEEESPM_2010}. 

Now, we analyze the computational complexity for solving problem (\ref{OptL}). Let $\epsilon$ be the iteration accuracy; then, the number of iterations is on the order of $\sqrt{2N_{\rm{IRS}}+5KM}\ln(1/\epsilon)$. For problem~(\ref{OptL}), there are 1 LMI constraint of size $N_{\rm{IRS}}$, $3KM+N_{\rm{IRS}}$ LMI constraint of size $1$, and $KM$ second-order cone constraints. Therefore, the complexity of solving problem (\ref{OptL}) is on the order of 
$\sqrt{\Delta_1}\zeta\ln(1/\epsilon)(\Delta_2+\zeta\Delta_3+\zeta^2)$,
where $\Delta_1=2N_{\rm{IRS}}+5KM$, $\Delta_2=N_{\rm{IRS}}^3+N_{\rm{IRS}}+12KM$, $\Delta_3=N_{\rm{IRS}}^2+N_{\rm{IRS}}+3KM$, and $\zeta=N_{\rm{IRS}}^2$. 

Finally, we summarize the proposed alteratively iterative optimization scheme in Algorithm~\ref{algorithm1}.
\begin{algorithm}[t]
	{\caption{The Proposed Alternatively Iterative Optimization Algorithm.}
		\label{algorithm1}
		Initialize  the reflection matrix ${\bf{\Phi}}^{(0)}$, iteration number $r=1$ and maximum iteration number $r_{\rm{max}}$.\\
		\Repeat{$r=r_{\rm{max}}$ {\rm{or} Convergence}}{
			Obtain the analog beamforming ${\bf{F}}^{(r)}$ according to (\ref{eq20}).\\
		    Initialize variables $t_{m,k}^{(0)}$, $b_{m,k}^{(0)}$, iteration number $r'=1$ and maximum iteration number $r'_{\rm{max}}$.\\		
			\Repeat{$r'=r'_{\rm{max}}$ {\rm{or} Convergence}}{
			Obtain $t_{m,k}^{\star}$, $b_{m,k}^{\star}$ and ${\bf{v}}_m[k]^{\star}$ by solving (\ref{OptH}).\\
		    Update variables $t_{m,k}^{(r')}\leftarrow t_{m,k}^{\star}$, $b_{m,k}^{(r')}\leftarrow b_{m,k}^{\star}$.\\
		    Update $r'\leftarrow r'+1$.\\}
		    Initialize variables ${t'}_{m,k}^{(0)}$, ${b'}_{m,k}^{(0)}$, iteration number $r{''}=1$ and maximum iteration number $r{''}_{\rm{max}}$.\\		
		    \Repeat{$r{''}=r{''}_{\rm{max}}$ {\rm{or} Convergence}}{
	    	Obtain ${t'}_{m,k}^{\star}$, ${b'}_{m,k}^{\star}$ and ${\bf{\Omega}}^{\star}$ by solving (\ref{OptL}).\\
	    	Update variables ${t'}_{m,k}^{(r{''})}\leftarrow {t'}_{m,k}^{\star}$, ${b'}_{m,k}^{(r{''})}\leftarrow {b'}_{m,k}^{\star}$.\\
	    	Update $r{''}\leftarrow r{''}+1$.\\}
    	    Obtain ${\bf{\Phi}}^\star$ according to ${\bf{\Omega}}^{\star}$.\\
    	    Update ${\bf{\Phi}}^{(r)}\leftarrow {\bf{\Phi}}^\star$.\\
			Update $r\leftarrow r+1$.\\
		}
	Obtain the analog beamforming ${\bf{\Phi}}^{(r)}$}, digital beamforming ${\bf{v}}_m[k]^{(r)}$ and reflection matrix ${\bf{\Phi}}^{(r)}$.
\end{algorithm} 
\section{Extension to Imperfect CSIs from IRS to Users}
Due to the mobility of users, it is difficult to obtain the perfect CSIs from the IRS to users. Therefore, in this section, we assume that there exists channel estimation error for the reflection links between IRS and users, namely ${{\bf{g}}}_m[k]={\tilde{\bf{g}}}_m[k]+\triangle{{\bf{g}}}_m[k]$, where ${\tilde{\bf{g}}}_m[k]$ denotes the estimated CSI and $\triangle{{\bf{g}}}_m[k]$ is the estimation error. Here, we assume that the estimation error is upper bounded with ${\tilde{\bf{g}}}_m[k]{\tilde{\bf{g}}}_m[k]^H\leq \varepsilon$. Next, we redesign the hybrid analog/digital beamforming and reflection matrix with imperfect CSI as follows
\begin{subequations}\label{OptG}
	\begin{align}
	\;\;\;\;\;\;\;\;\;\;\;\;\;\;\;\;\;\;\;\;\;\;\;\;\;\;\;\;\;\;\;\;&\underset{\left\{{\bf\Phi},{\bf{F}},{\bf{v}}_m[k]\right\}}{\rm{max}}\;\sum\nolimits_{m=1}^{M}\alpha_mR_m \label{OptG0}\\
	{\rm{s.t.}}&\;|\phi_i|=1,i\in \{1,\cdots,N_{\rm{IRS}}\}, \label{OptG1}\\
	&\;\sum\nolimits_{M=1}^M\sum\nolimits_{k=1}^K||{\bf{F}}{\bf{v}}_m[k]||^2\leq P_{\rm{max}},\label{OptG2}\\
	&\;{{\bf{g}}}_m[k]={\tilde{\bf{g}}}_m[k]+\triangle{{\bf{g}}}_m[k],\label{OptG3}\\
	&\;{\tilde{\bf{g}}}_m[k]{\tilde{\bf{g}}}_m[k]^H\leq \varepsilon,\label{OptG4}\\
	&\;{\bf{F}}(i,j)=1/\sqrt{N_{\rm{TX}}},i\in \{1,\cdots,N_{\rm{TX}}\}, j\in\{1,\cdots,N_{\rm{RF}}\}. \label{OptG5}
	\end{align}
\end{subequations} 
Similar to the proposed Algorithm~\ref{algorithm1}, we design an alternatively iterative optimization scheme to solve it. 
\subsection{Optimization of  ${\bf{F}}$ and ${\bf{v}}_m[k]$ under Fixed ${\bf\Phi}$}
In fact, it is extremely difficult to design analog beamforming ${\bf{F}}$ under channel estimation error, and thus, we ignore the estimation error and define ${{\bf{g}}}_m[k] \triangleq {\tilde{\bf{g}}}_m[k]$ for simplicity, namely ${\hat{\bf{h}}}_m[k]=u_m[k]{\tilde{\bf{g}}}_m[k]{\bf\Phi}{{\bf{H}}}[k]$. In this way, we can adopt the same scheme used in Section~\ref{A1} to obtain the  analog beamforming ${\bf{F}}$. Next, we directly solve the digital beamforming ${\bf{v}}_m[k]$. After obtaining ${\bf{F}}$, we define $u_m[k]{{\bf{g}}}_m[k]{\bf\Phi}{{\bf{H}}}[k]{\bf{F}}\triangleq {{\bf{g}}}_m[k]{\bf{\Xi}}_m[k]$, where ${\bf{\Xi}}_m[k]=u_m[k]{\bf\Phi}{{\bf{H}}}[k]{\bf{F}}$. Thus, the received signal of the $m$th user on the $k$th subcarrier can be expressed as
\begin{equation}\label{eq32}
\begin{aligned}
y_m[k]=({\tilde{\bf{g}}}_m[k]+\triangle{{\bf{g}}}_m[k]){\bf{\Xi}}_m[k]{\bf{v}}_m[k]x_m[k]
+\sum\nolimits_{j\neq m}^{M}({\tilde{\bf{g}}}_m[k]+\triangle{{\bf{g}}}_m[k]){\bf{\Xi}}_m[k]{\bf{v}}_j[k]x_j[k]+n_m[k],
\end{aligned}
\end{equation}
and the achievable rate is given by
\begin{eqnarray}\label{eq33}
R_m[k]=\frac{B}{K}\log\left(1\!+\!\frac{\left|({\tilde{\bf{g}}}_m[k]\!+\!\triangle{{\bf{g}}}_m[k]){\bf{\Xi}}_m[k]{\bf{v}}_m[k]\right|^2}{\sum\nolimits_{j\neq m}^{M}\Upsilon_{m}[k]\!+\!BN_0/K}\right),
\end{eqnarray}
where $\Upsilon_{m}[k]=\left|({\tilde{\bf{g}}}_m[k]\!+\!\triangle{{\bf{g}}}_m[k]){\bf{\Xi}}_m[k]{\bf{v}}_j[k]\right|^2$. By introducing an auxiliary variable $\tau_{m,k}$, the original problem (\ref{OptG}) can be recast~as
\begin{subequations}\label{OptK}
	\begin{align}
	\;\;\;\;\;\;\;\;\;\;\;\;\;\;\;\;\;\;\;\;\;\;\;\;\;\;\;\;\;\;\;\;\;\;\;\;\;\;\;&\underset{\left\{{\bf{v}}_m[k]\right\}}{\rm{max}}\;\sum_{m=1}^{M}\sum_{k=1}^{K}a_m\log\left(1\!+\tau_{m,k}\right) \label{OptK0}\\
	{\rm{s.t.}}	&\;\tau_{m,k}\leq \frac{\left|({\tilde{\bf{g}}}_m[k]\!+\!\triangle{{\bf{g}}}_m[k]){\bf{\Xi}}_m[k]{\bf{v}}_m[k]\right|^2}{\sum\nolimits_{j\neq m}^{M}\Upsilon_{m}[k]\!+\!BN_0/K},\label{OptK1}\\
	&\;\;\triangle{{\bf{g}}}_m[k]\triangle{{\bf{g}}}_m[k]^H\leq \varepsilon,\label{OptK2}\\
	&\;\sum\nolimits_{M=1}^M\sum\nolimits_{k=1}^K||{\bf{F}}{\bf{v}}_m[k]||^2\leq P_{\rm{max}}.
	\end{align}
\end{subequations} 

By comparing the weighted sum rate maximization problem under perfect CSI (\ref{OptE}) and imperfect CSI (\ref{OptK}), one can observe that the original scheme used to solve (\ref{OptE}) can not be directly used to solve (\ref{OptK}) due to the uncertain term  $\triangle{{\bf{g}}}_m[k]$. To cope with~(\ref{OptK}), we first give the following  Lemma~\ref{lemma1} referred to as $\mathcal{S}$-Procedure~\cite{Zhou_IEEETWC_2017}.
\begin{lemma}\label{lemma1}
	Define the function
	\begin{eqnarray}
	f_i({\bf{x}})={\bf{x}}{\bf{Q}}_i{\bf{x}}^H+2{\rm{Re}}\{{\bf{p}}_i{\bf{x}}^H\}+e_i, i\in\{1,2\},
	\end{eqnarray}
	where ${\bf{x}}\in{\mathbb{C}}^{1\times O}$, ${\bf{Q}}_i\in{\mathbb{C}}^{O\times O}$, ${\bf{p}}_i\in{\mathbb{C}}^{1\times O}$, and $e_i\in\mathbb{R}$ with $O$ representing any integer, and thus, the expression $f_1({\bf{x}})\leq 0\Rightarrow f_2({\bf{x}})\leq 0$ holds if and only if there exists a $\beta$ satisfying
	\begin{eqnarray}\label{C121}
	\beta\left[ \begin{array}{ccc}
	{\bf{Q}}_1 & {\bf{p}}_1^H \\
	{\bf{p}}_1 & e_1
	\end{array} 
	\right ]-\left[ \begin{array}{ccc}
	{\bf{Q}}_2 & {\bf{p}}_2^H \\
	{\bf{p}}_2 & e_2
	\end{array} 
	\right ]\succeq {\bf{0}}.
	\end{eqnarray}  
\end{lemma} 

Next, we have 
\begin{eqnarray}\label{eq37}
\begin{aligned}
	&\left|({\tilde{\bf{g}}}_m[k]\!+\!\triangle{{\bf{g}}}_m[k]){\bf{\Xi}}_m[k]{\bf{v}}_m[k]\right|^2\\
	=&\left(({\tilde{\bf{g}}}_m[k]\!+\!\triangle{{\bf{g}}}_m[k]){\bf{\Xi}}_m[k]{\bf{v}}_m[k]\right)\left(({\tilde{\bf{g}}}_m[k]\!+\!\triangle{{\bf{g}}}_m[k]){\bf{\Xi}}_m[k]{\bf{v}}_m[k]\right)^H\\
	=&\triangle{{\bf{g}}}_m[k]{\bf{\Xi}}_m[k]{\bf{V}}_m[k]{\bf{\Xi}}_m[k]^H\triangle{{\bf{g}}}_m[k]^H
	+2{\rm{Re}}({\tilde{\bf{g}}}_m[k]{\bf{\Xi}}_m[k]{\bf{V}}_m[k]{\bf{\Xi}}_m[k]^H\triangle{{\bf{g}}}_m[k]^H)\\
	&+{\tilde{\bf{g}}}_m[k]{\bf{\Xi}}_m[k]{\bf{V}}_m[k]{\bf{\Xi}}_m[k]^H{\tilde{\bf{g}}}_m[k]^H.
\end{aligned}
\end{eqnarray}
Similarly, $\sum_{j\neq m}\Upsilon_{m}[k]$ can be expressed as
\begin{eqnarray}\label{eq38}
\begin{aligned}
&\sum\nolimits_{j\neq m}\Upsilon_{m}[k]=\sum\nolimits_{j\neq m}\left|({\tilde{\bf{g}}}_m[k]\!+\!\triangle{{\bf{g}}}_m[k]){\bf{\Xi}}_m[k]{\bf{v}}_j[k]\right|^2\\
=&\triangle{{\bf{g}}}_m[k]{\bf{\Xi}}_m[k]\left(\sum\nolimits_{j\neq m}{\bf{V}}_j[k]\right){\bf{\Xi}}_m[k]^H\triangle{{\bf{g}}}_m[k]^H
+2{\rm{Re}}({\tilde{\bf{g}}}_m[k]{\bf{\Xi}}_m[k]\left(\sum\nolimits_{j\neq m}{\bf{V}}_j[k]\right){\bf{\Xi}}_m[k]^H\triangle{{\bf{g}}}_m[k]^H)\\
&+{\tilde{\bf{g}}}_m[k]{\bf{\Xi}}_m[k]\left(\sum\nolimits_{j\neq m}{\bf{V}}_j[k]\right){\bf{\Xi}}_m[k]^H{\tilde{\bf{g}}}_m[k]^H.
\end{aligned}
\end{eqnarray}
Next, we transform (\ref{OptK1}) into the following three constraints
\begin{subequations}
	\begin{align}
	\;\;\;\;\;\;\;\;\;\;\;\;\;\;\;\;\;\;\;\;\;\;\;\;\;\;\;\;\;\;\;\;\;\;\;\;\;\;\;&\mu_{m,k}\leq \left|({\tilde{\bf{g}}}_m[k]\!+\!\triangle{{\bf{g}}}_m[k]){\bf{\Xi}}_m[k]{\bf{v}}_m[k]\right|^2,\label{eq39A}\\
	&\tau_{m,k}\theta_{m,k}\leq \mu_{m,k},\label{eq39B}\\
	&\theta_{m,k}\geq \sum\nolimits_{j\neq m}^{M}\Upsilon_{m}[k]\!+\!BN_0/K. \label{eq39C}
	\end{align}
\end{subequations}	

Combining (\ref{eq37}), (\ref{OptK2}) and Lemma~\ref{lemma1}, (\ref{eq39A}) can be transformed into the following convex LMI constraint
\begin{eqnarray}\label{eq40}
\left[ \begin{array}{ccc}
\beta_{m,k}{\bf{I}}\!+\!\hat{{\bf{V}}}_m[k] & ({\tilde{\bf{g}}}_m[k]\hat{{\bf{V}}}_m[k])^H \\
{\tilde{\bf{g}}}_m[k]\hat{{\bf{V}}}_m[k] & \varepsilon\!+\!{\tilde{\bf{g}}}_m[k]\hat{{\bf{V}}}_m[k]{\tilde{\bf{g}}}_m[k]^H-\mu_{m,k}
\end{array} 
\right ]\!\succeq\! {\bf{0}},
\end{eqnarray}  
where $\hat{{\bf{V}}}_m[k]={\bf{\Xi}}_m[k]{\bf{V}}_m[k]{\bf{\Xi}}_m[k]^H$. The non-convex constraint (\ref{eq39B}) can be expressed as the following convex constraint
\begin{eqnarray}\label{eq41}
\frac{\tau_{m,k}^{(n)}}{2\theta_{m,k}^{(n)}}\theta_{m,k}^2+\frac{\theta_{m,k}^{(n)}}{2\tau_{m,k}^{(n)}}\tau_{m,k}^2\leq \mu_{m,k},
\end{eqnarray}
where $\tau_{m,k}^{(n)}$ and $\theta_{m,k}^{(n)}$ are the values of $\tau_{m,k}$ and $\theta_{m,k}$ at the $n$th iteration, respectively. In addition, combining~(\ref{eq38}), (\ref{OptK2}) and Lemma~\ref{lemma1}, (\ref{eq39C}) can be transformed into the following convex LMI constraint
\begin{eqnarray}\label{eq42}
\left[ \begin{array}{ccc}
\beta_{m,k}{\bf{I}}\!-\!\check{{\bf{V}}}_m[k] & -({\tilde{\bf{g}}}_m[k]\check{{\bf{V}}}_m[k])^H \\
-{\tilde{\bf{g}}}_m[k]\check{{\bf{V}}}_m[k] & \varepsilon\!-\!{\tilde{\bf{g}}}_m[k]\check{{\bf{V}}}_m[k]{\tilde{\bf{g}}}_m[k]^H+\theta_{m,k}-BN_0/K
\end{array} 
\right ]\!\succeq\! {\bf{0}},
\end{eqnarray}  
where $\check{{\bf{V}}}_m[k]={\bf{\Xi}}_m[k]\left(\sum\nolimits_{j\neq m}{\bf{V}}_j[k]\right){\bf{\Xi}}_m[k]^H$.

Finally, we transform (\ref{OptM}) into the following SDP problem
\begin{subequations}\label{OptU}
	\begin{align}
	\;\;\;\;\;\;\;\;\;\;\;\;\;\;\;\;\;\;\;\;\;\;\;\;\;\;\;\;\;\;\;\;\;\;\;\;\;&\underset{\left\{\tau_{m,k},\mu_{m,k},\theta_{m,k},\left\{{\bf{V}}_m[k]\right\}\right\}}{\rm{max}}\;\sum_{m=1}^{M}\sum_{k=1}^{K}a_m\log\left(1\!+\tau_{m,k}\right) \label{OptU0}\\
	{\rm{s.t.}}	&\;{\rm{(\ref{eq40}),(\ref{eq41}),(\ref{eq42})}},\label{OptU1}\\
		&\;\sum\nolimits_{M=1}^M\sum\nolimits_{k=1}^K{\rm{Tr}}({\bf{F}}^H{\bf{F}}{\bf{V}}_m[k])\leq P_{\rm{max}},\label{OptU2}\\
	&\;{\rm{rank}}({\bf{V}}_m[k])=1, {\bf{V}}_m[k]\succeq 0. \label{OptU3}
	\end{align}
\end{subequations} 

It is obvious that (\ref{OptU}) can be solved with the convex optimization toolbox such as CVX by removing the rank-one constraint. Likewise, the optimal digital beamforming matrix ${\bf{V}}_m[k]$ is obtained by iteratively solving problem (\ref{OptU}).

To handle problem (\ref{OptU}), we remove the rank-one constraint ${\rm{rank}}({\bf{V}}_m[k])=1$.  If problem (\ref{OptU}) does not yield a rank-one solution, namely ${\rm{rank}}({\bf{V}}_m[k])>1$, the Gaussian randomization technique is  used to obtain a rank-one solution~\cite{Luo_IEEESPM_2010}.

Now, we analyze the computational complexity for solving problem (\ref{OptU}). Let $\epsilon$ be the iteration accuracy; then, the number of iterations is on the order of $\sqrt{3KMN_{\rm{RF}}+7KM}\ln(1/\epsilon)$\cite{Chi_IEEETSP_2014}. For problem~(\ref{OptU}), there are $2KM$ LMI constraints of size $N_{\rm{IRS}}+1$, $KM$ LMI constraints of size $N_{\rm{IRS}}$, $3KM+1$ LMI constraints of size 1, and $KM$ second-order cone constraints. Therefore, the complexity of solving problem (\ref{OptU}) is on the order of 
$\sqrt{\Delta_1}\zeta\ln(1/\epsilon)(\Delta_2+\zeta\Delta_3+\zeta^2)$,
where $\Delta_1=3KMN_{\rm{RF}}+7KM$, $\Delta_2=KM(2(N_{\rm{RF}}+1)^3+N_{\rm{RF}}^3+11)+1$, $\Delta_3=2KM(N_{\rm{RF}}+1)^2+KMN_{\rm{RF}}^2+3KM+1$, and $\zeta=KMN_{\rm{RF}}^2$.
\subsection{Optimization of $\bf{\Phi}$ under Fixed ${\bf{F}}$ and ${\bf{V}}_m[k]$}
After obtaining the  analog beamforming ${\bf{F}}$ and digital beamforming ${\bf{v}}_m[k]$, we reconsider the reflection matrix $\bf{\Phi}$ with imperfect CSIs. According to the definition in Section~\ref{B}, we have
\begin{eqnarray}
	u_m[k]{{\bf{g}}}_m[k]{\bf\Phi}{{\bf{H}}}[k]{\bf{F}}{\bf{v}}_m[k]\triangleq {{\bf{g}}}_m[k]{\bf\Phi}{{\bf{z}}}_m[k].
\end{eqnarray}

Next, we define ${{\bf{Z}}}_m[k]\triangleq{\rm{diag}}({{\bf{z}}}_m[k])$, and thus, we have ${{\bf{g}}}_m[k]{\bf\Phi}{{\bf{z}}}_m[k]={{\bf{g}}}_m[k]{{\bf{Z}}}_m[k]{\bm{\phi}}$. After that, the achievable rate can be expressed as
\begin{eqnarray}\label{eq45}
R_m[k]=\frac{B}{K}\log\left(1\!+\!\frac{\left|({\tilde{\bf{g}}}_m[k]\!+\!\triangle{{\bf{g}}}_m[k]){{\bf{Z}}}_m[k]{\bm{\phi}}\right|^2}{\sum\nolimits_{j\neq m}^{M}\Upsilon'_{m}[k]\!+\!BN_0/K}\right),
\end{eqnarray}
where $\Upsilon'_{m}[k]=\left|({\tilde{\bf{g}}}_m[k]\!+\!\triangle{{\bf{g}}}_m[k]){{\bf{Z}}}_j[k]{\bm{\phi}}\right|^2$. It can be found that (\ref{eq45}) has a similar expression with (\ref{eq33}), and thus, the same scheme can be used to solve the reflection matrix $\bf{\Phi}$. Therefore, we omit the details and directly formulate the following SDR problem as
\begin{subequations}\label{OptY}
	\begin{align}
	\;\;\;\;\;\;\;\;\;\;\;\;\;\;\;\;\;\;\;\;\;\;\;\;\;\;\;\;\;\;\;\;&\underset{\left\{\tau'_{m,k},\mu'_{m,k},\theta'_{m,k},\bf{\Omega}\right\}}{\rm{max}}\;\sum_{m=1}^{M}\sum_{k=1}^{K}a_m\log\left(1\!+\tau'_{m,k}\right) \label{OptY0}\\
	{\rm{s.t.}}	&\;\left[ \begin{array}{ccc}
	\beta'_{m,k}{\bf{I}}\!+\!\hat{{\bf{\Omega}}}_m[k] & ({\tilde{\bf{g}}}_m[k]\hat{{\bf{\Omega}}}_m[k])^H \\
	{\tilde{\bf{g}}}_m[k]\hat{{\bf{\Omega}}}_m[k] & \varepsilon\!+\!{\tilde{\bf{g}}}_m[k]\hat{{\bf{\Omega}}}_m[k]{\tilde{\bf{g}}}_m[k]^H-\mu'_{m,k}
	\end{array} 
	\right ]\!\succeq\! {\bf{0}},\label{OptY1}\\
	&\frac{{\tau'}_{m,k}^{(n)}}{2{\theta'}_{m,k}^{(n)}}{\theta'}_{m,k}^2+\frac{{\theta'}_{m,k}^{(n)}}{2{\tau'}_{m,k}^{(n)}}{\tau'}_{m,k}^2\leq \mu'_{m,k},\label{OptY2}\\
	&\left[ \begin{array}{ccc}
	\beta'_{m,k}{\bf{I}}\!-\!\check{{\bf{\Omega}}}_m[k] & -({\tilde{\bf{g}}}_m[k]\check{{\bf{\Omega}}}_m[k])^H \\
	-{\tilde{\bf{g}}}_m[k]\check{{\bf{\Omega}}}_m[k] & \varepsilon\!-\!{\tilde{\bf{g}}}_m[k]\check{{\bf{\Omega}}}_m[k]{\tilde{\bf{g}}}_m[k]^H+\theta'_{m,k}-BN_0/K
	\end{array} 
	\right ]\!\succeq\! {\bf{0}},\\
	&\;{\bf{\Omega}}(i,i)=1, {\bf{\Omega}}\succeq 0, \label{OptY4}
	\end{align}
\end{subequations} 
where $\tau_{m,k}$, $\mu'_{m,k}$ and $\theta'_{m,k}$ are the introduced auxiliary variables,  $\hat{{\bf{\Omega}}}_m[k]={\bf{Z}}_m[k]{\bf{\Omega}}{\bf{Z}}_m[k]^H$, and $\check{{\bf{\Omega}}}_m[k]=\sum\nolimits_{j\neq m}{\bf{\Xi}}_j[k]{\bf{\Omega}}{\bf{\Xi}}_j[k]^H$. We can use the CVX toolbox to solve (\ref{OptY}). Likewise, if problem (\ref{OptL}) does not yield a rank-one solution, namely ${\rm{rank}}({\bf{\Omega}}^\star)\neq 1$. The Gaussian randomization technique can be  used to obtain a rank-one solution~\cite{Luo_IEEESPM_2010}. 

Now, we analyze the computational complexity for solving problem (\ref{OptY}). Let $\epsilon$ be the iteration accuracy; then, the number of iterations is on the order of $\sqrt{KMN_{\rm{IRS}}+7KM}\ln(1/\epsilon)$. For problem~(\ref{OptY}), there are $2KM$ LMI constraints of size $N_{\rm{IRS}}+1$,  1 LMI constraint of size $N_{\rm{IRS}}$, $4KM$ LMI constraints of size $1$, and $KM$ second-order cone constraints. Therefore, the complexity of solving problem (\ref{OptY}) is on the order of 
$\sqrt{\Delta_1}\zeta\ln[1/\epsilon)(\Delta_2+12KM+1+\zeta\Delta_2+\zeta^2]$,
where $\Delta_1=KMN_{\rm{IRS}}+7KM$, $\Delta_2=2KM(N_{\rm{IRS}}+1)^3+N_{\rm{IRS}}^3+13KM$, $\Delta_3=2KM(N_{\rm{RF}}+1)^2+N_{\rm{RF}}^2+4KM$, and $\zeta=N_{IRS}^2$.
\section{Numerical Results}
In this section, simulation results are presented to evaluate the performance of the proposed schemes in IRS-aided THz MIMO-OFDMA systems. Due to the severe pathloss in THz, we consider a short distance communication scenario as shown in Fig.~\ref{Systemfigure2}, where  users are located within a circle with $1.5$ m radius. The AoD/AOA follows the uniform distribution within $[-\pi/2,\pi/2]$, and the antenna spacing is assumed to  be  half wavelength.  The default simulation parameters are listed in Table~\ref{Table I}, and they are used in simulation unless otherwise specified.

\begin{figure}[t]
	\begin{center}
		\includegraphics[width=10cm,height=6cm]{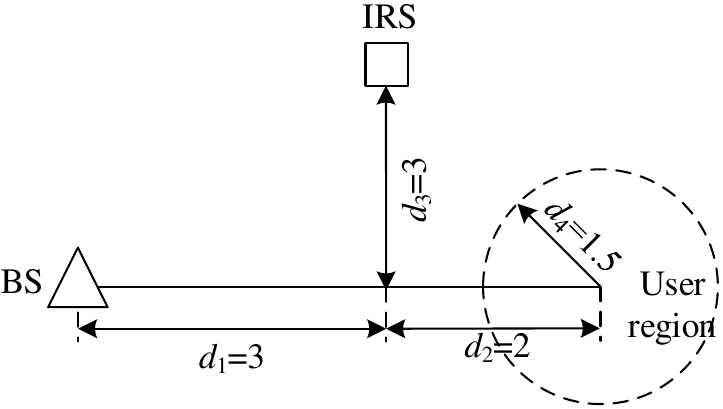}
		\caption{The location distribution in the IRS-aided system.}
		\label{Systemfigure2}
	\end{center}
\end{figure}  
\begin{table} [t]
	\caption{Default Parameters Used in Simulations.} 
	\renewcommand{\arraystretch}{0.8}
	\label{Table I} 
	\centering
	\begin{tabular}{l|r} 
		\hline  
		\bfseries Parameters & \bfseries Value \\ [0.5ex] 
		\hline\hline
		Number of antennas  &$N_{\rm{TX}}=64$ \\
		\hline
		Number of RF chains  &$N_{\rm{RF}}=4$ \\
		\hline
		Number of reflection elements &$N_{\rm{IRS}}=4$\\
		\hline
		Central frequency & $f_c=340$ [GHz]\\
		\hline
		Bandwidth &$B=20$ [GHz]\\
		\hline
		Number of subcarriers & $K=16$\\
		\hline
		Number of users & $M=2$ \\
		\hline
		Transmit antenna gain & $G_t=4\!+\!20\log_{10}(\sqrt{N_{\rm{TX}}})$  \\
		\hline
		Receive antenna gain & $G_r=1$\\
		\hline  
		Absorption coefficient & $0.0033/m$\\
		\hline
		Speed of light & $c=3\times 10^8$ \\ 
	    \hline
		Noise variance& $N_0=-174$ [dBm/Hz] \\
		\hline 
	\end{tabular}
\end{table}
\begin{figure}[t]
	\begin{center}
		\includegraphics[width=9cm,height=7cm]{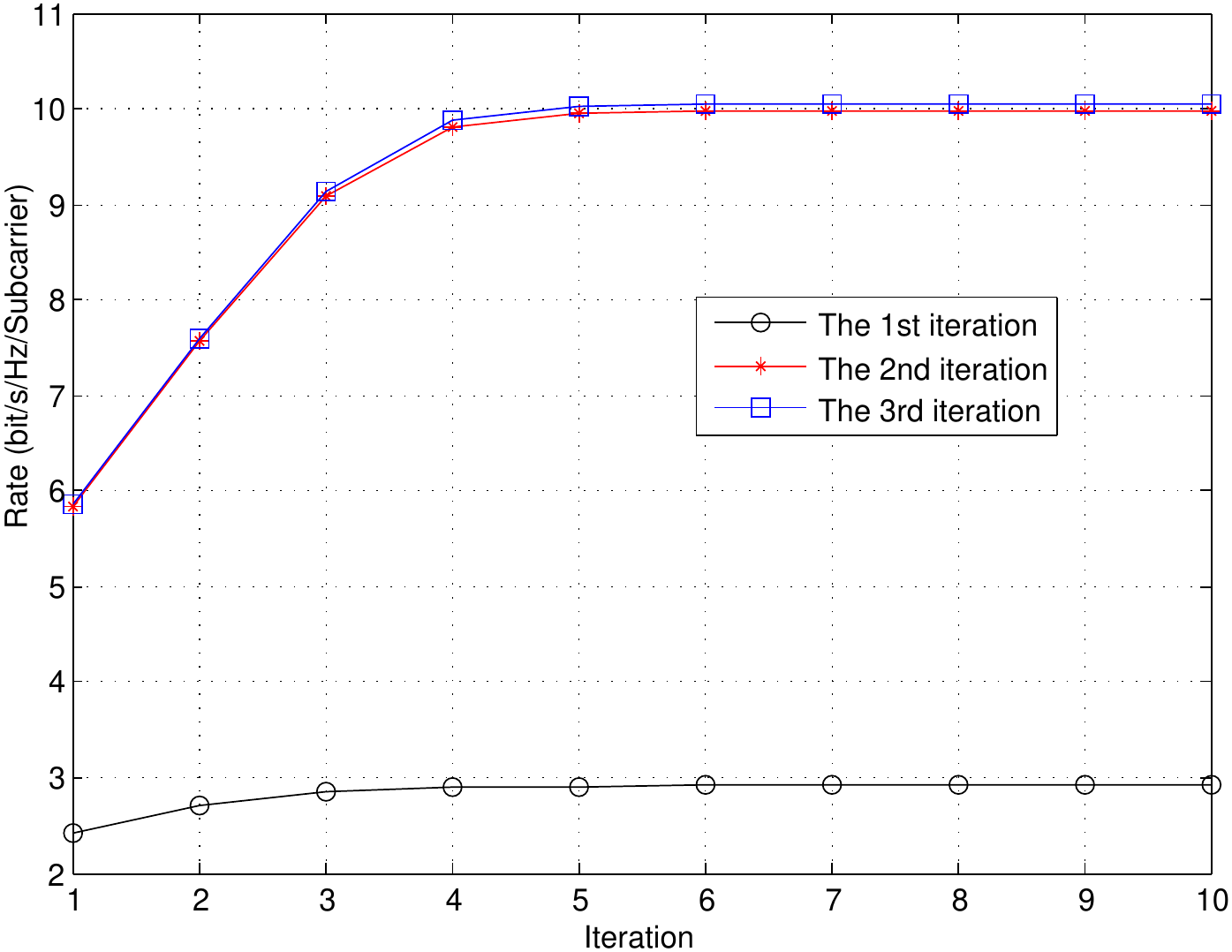}
		\caption{The rate versus iteration for solving the digital beamforming.}
		\label{figure1}
	\end{center}
\end{figure}
\begin{figure}[t]
	\begin{center}
		\includegraphics[width=9cm,height=7cm]{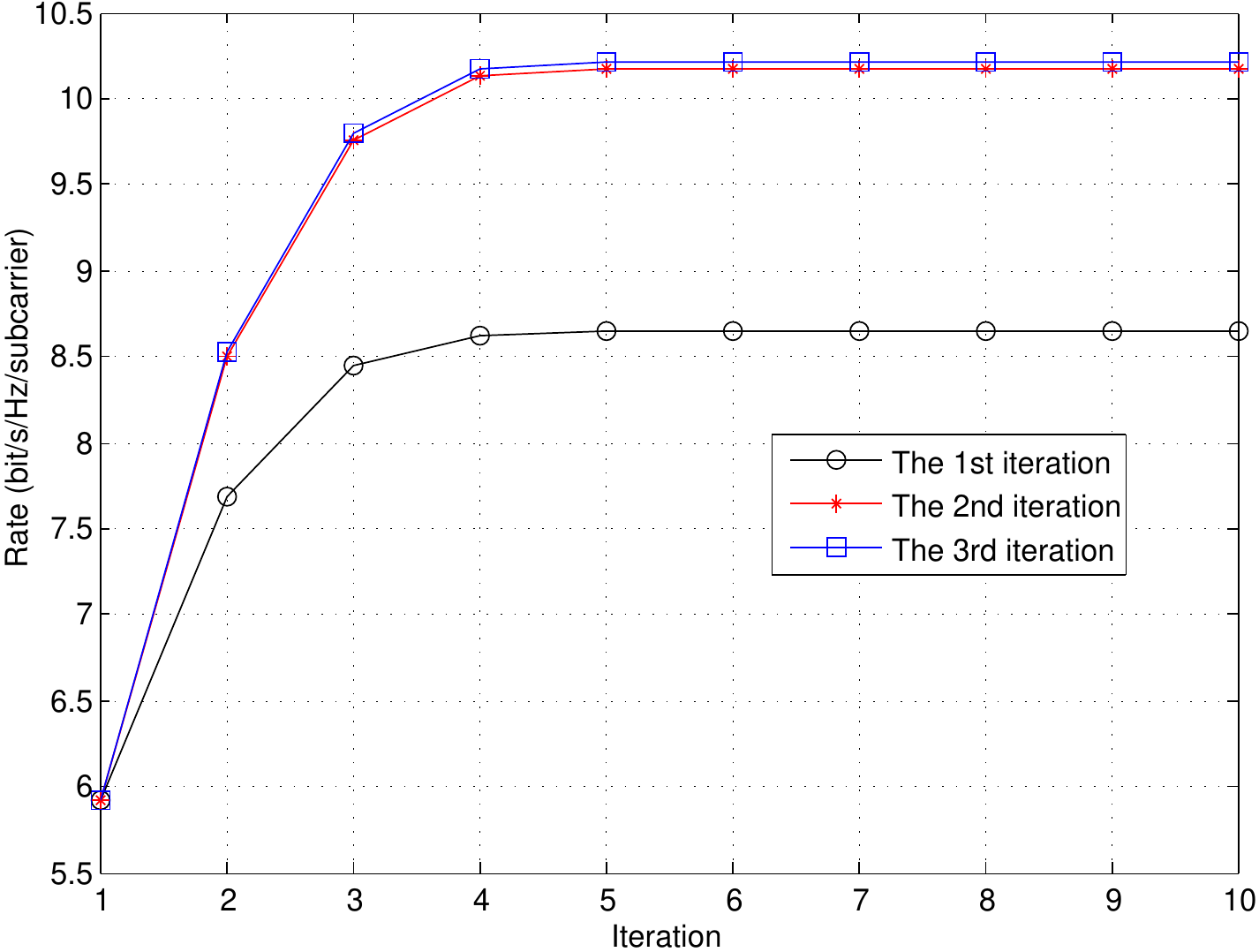}
		\caption{The rate versus iteration for solving the reflection matrix.}
		\label{figure2}
	\end{center}
\end{figure}
\begin{figure}[t]
	\begin{center}
		\includegraphics[width=9cm,height=7cm]{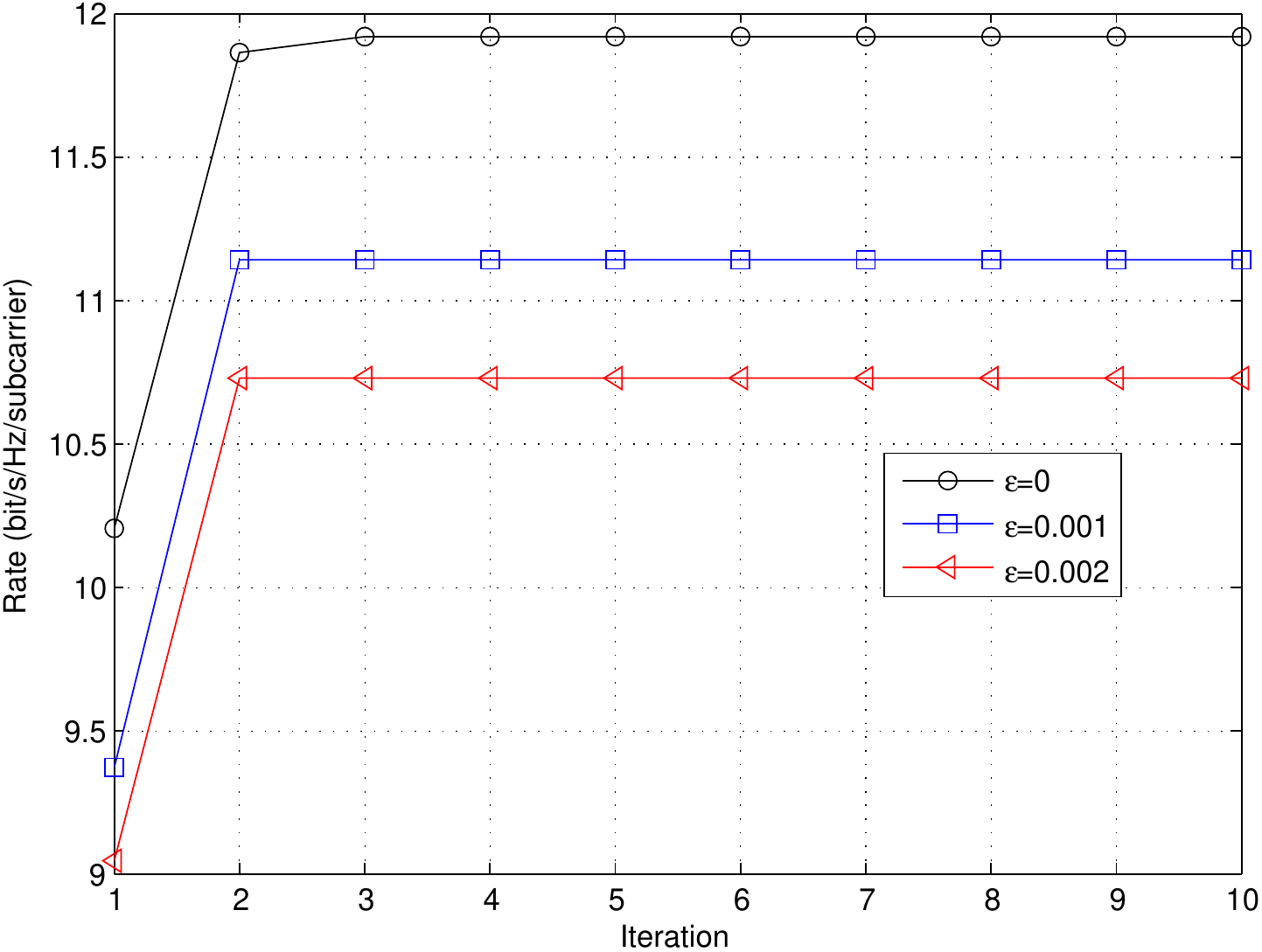}
		\caption{The rate versus iteration for the proposed Algorithm 1.}
		\label{figure3}
	\end{center}
\end{figure}
\begin{figure}[t]
	\begin{center}
		\includegraphics[width=9cm,height=7cm]{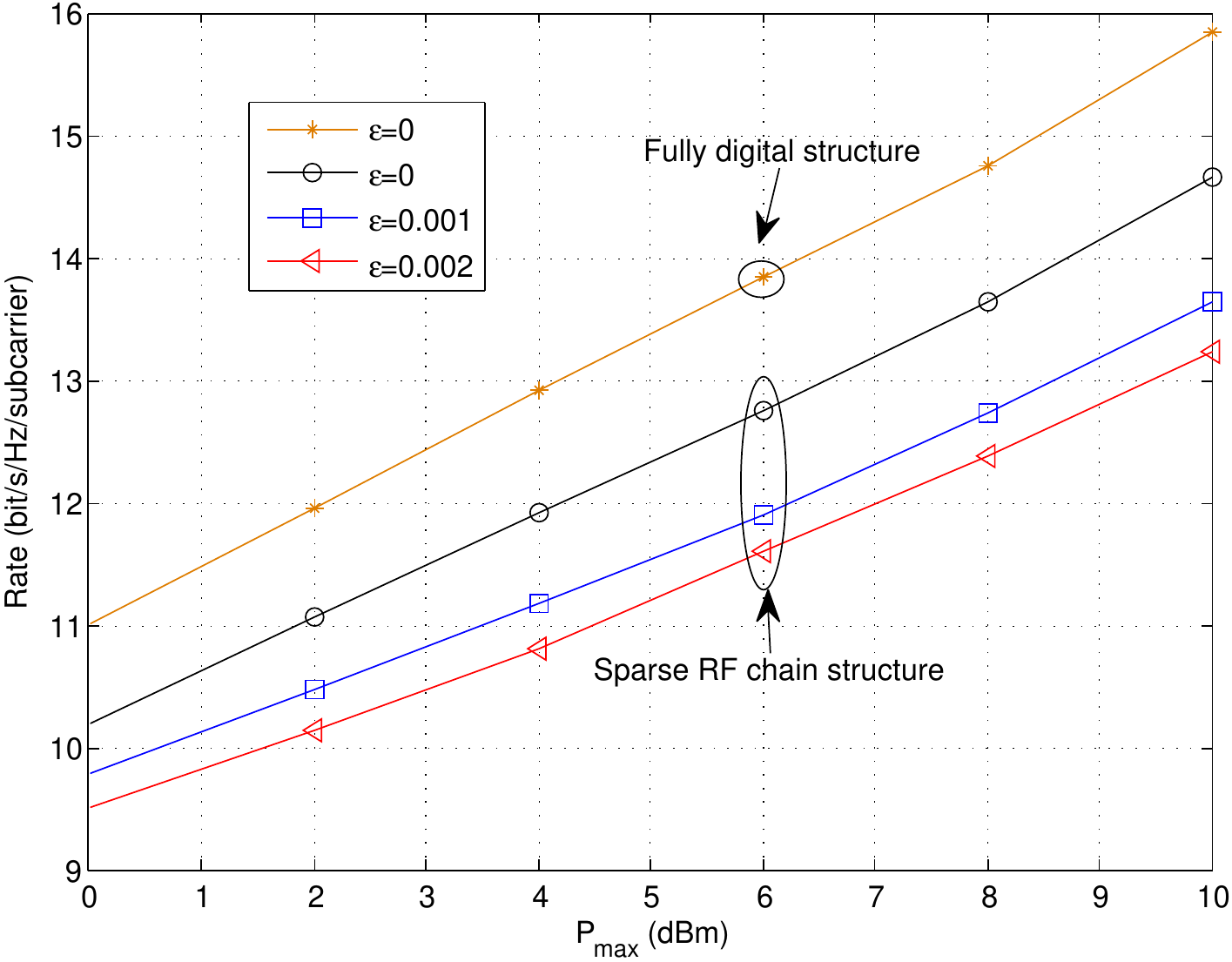}
		\caption{The rate versus the allowable maximum transmit power.}
		\label{figure4}
	\end{center}
\end{figure}
\begin{figure}[t]
	\begin{center}
		\includegraphics[width=9cm,height=7cm]{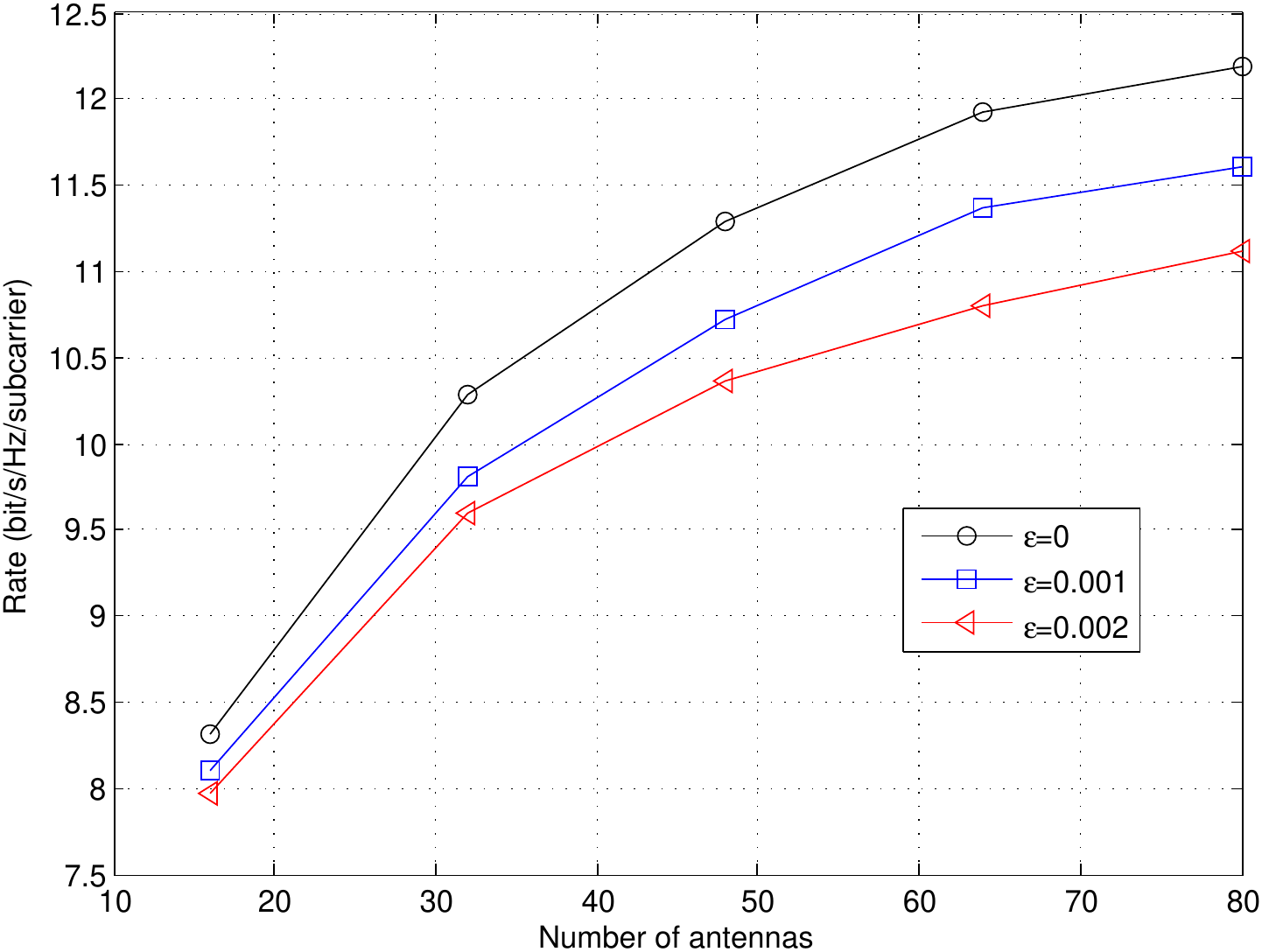}
		\caption{The rate versus the number of antennas.}
		\label{figure5}
	\end{center}
\end{figure}
\begin{figure}[t]
	\begin{center}
		\includegraphics[width=9cm,height=7cm]{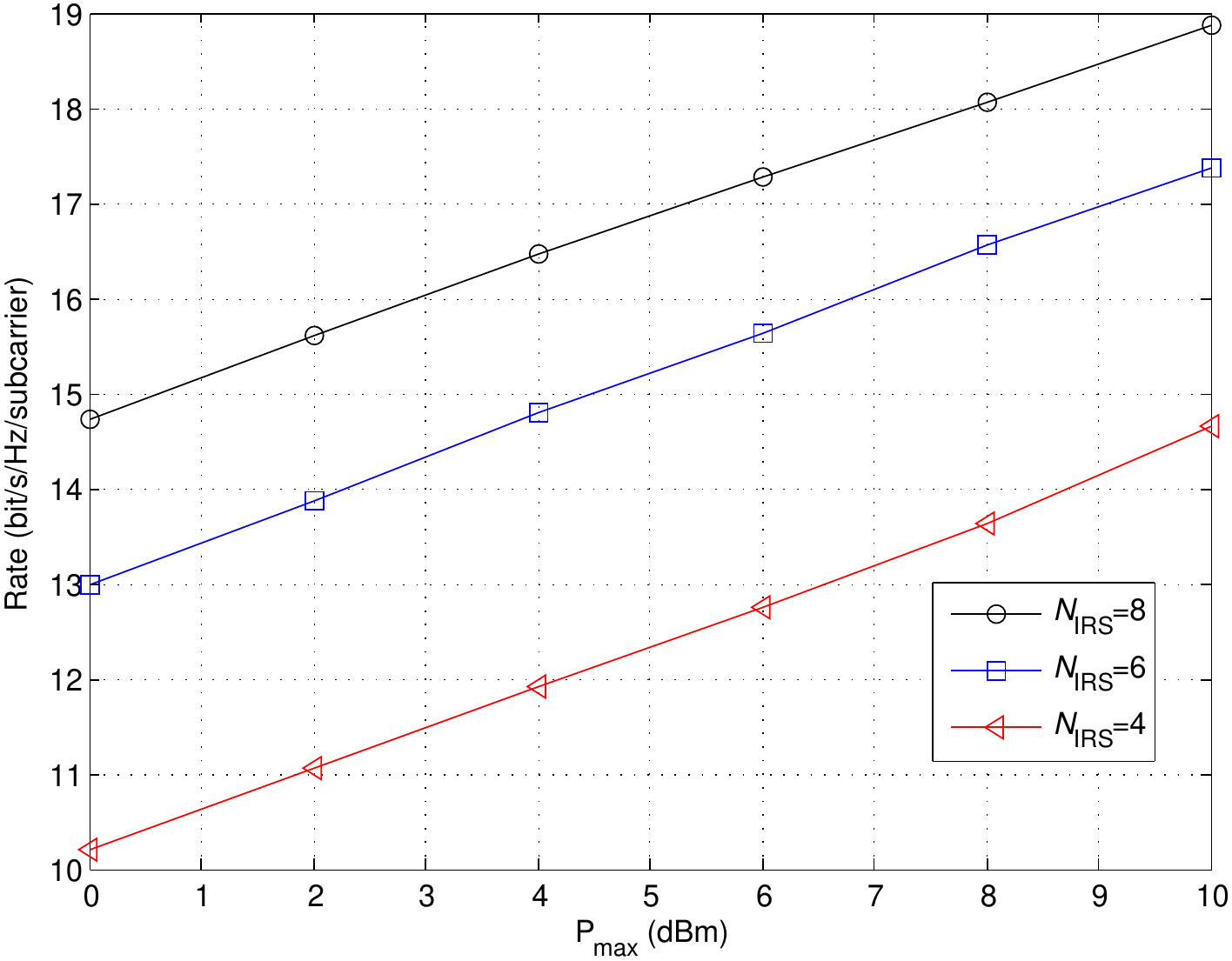}
		\caption{The rate versus the allowable maximum transmit power under different numbers of  IRS reflection elements.}
		\label{figure6}
	\end{center}
\end{figure}
\begin{figure}[t]
	\begin{center}
		\includegraphics[width=9cm,height=7cm]{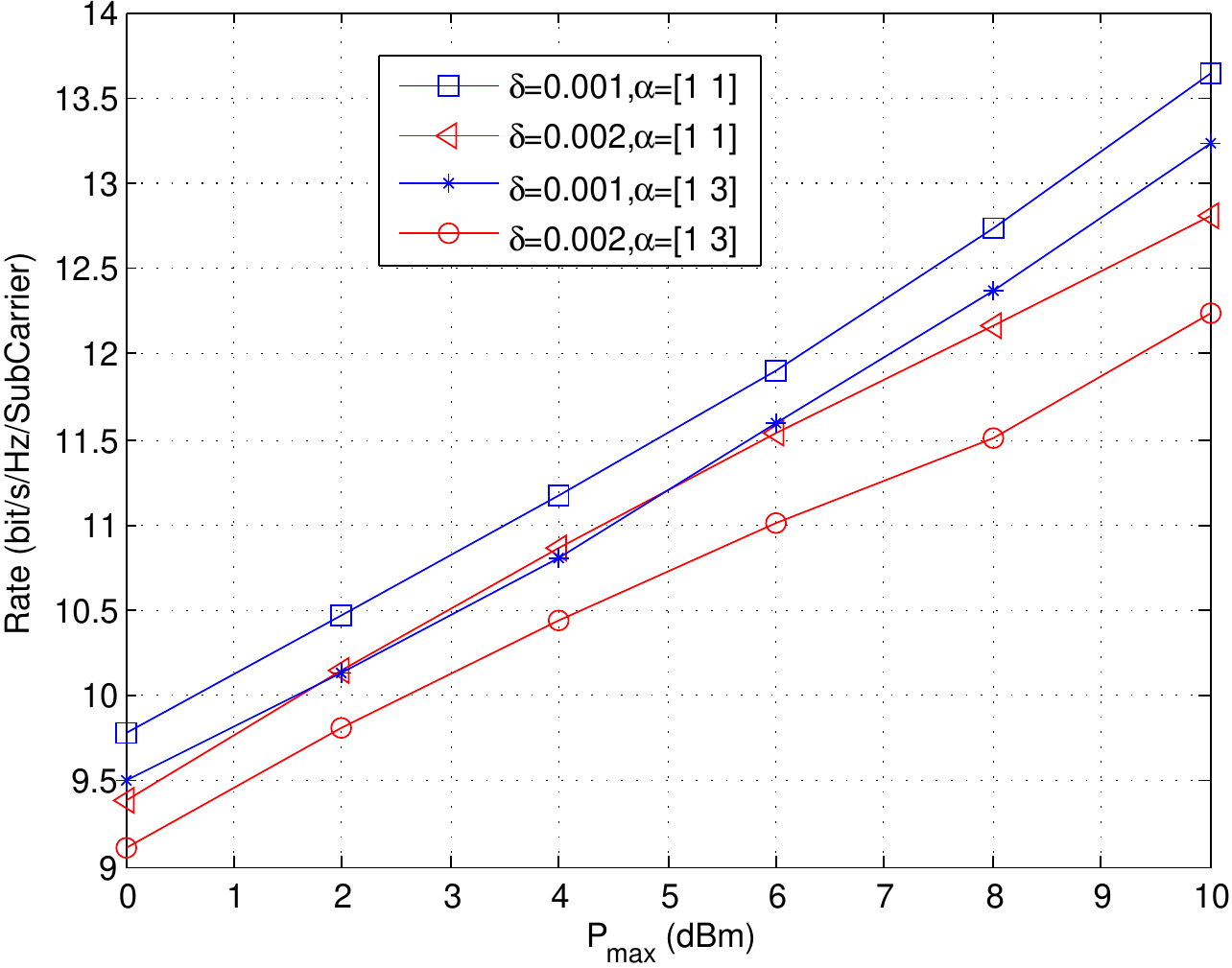}
		\caption{The rate versus the allowable maximum transmit power with different weights.}
		\label{figure7}
	\end{center}
\end{figure}
\begin{figure}[t]
	\begin{center}
		\includegraphics[width=9cm,height=7cm]{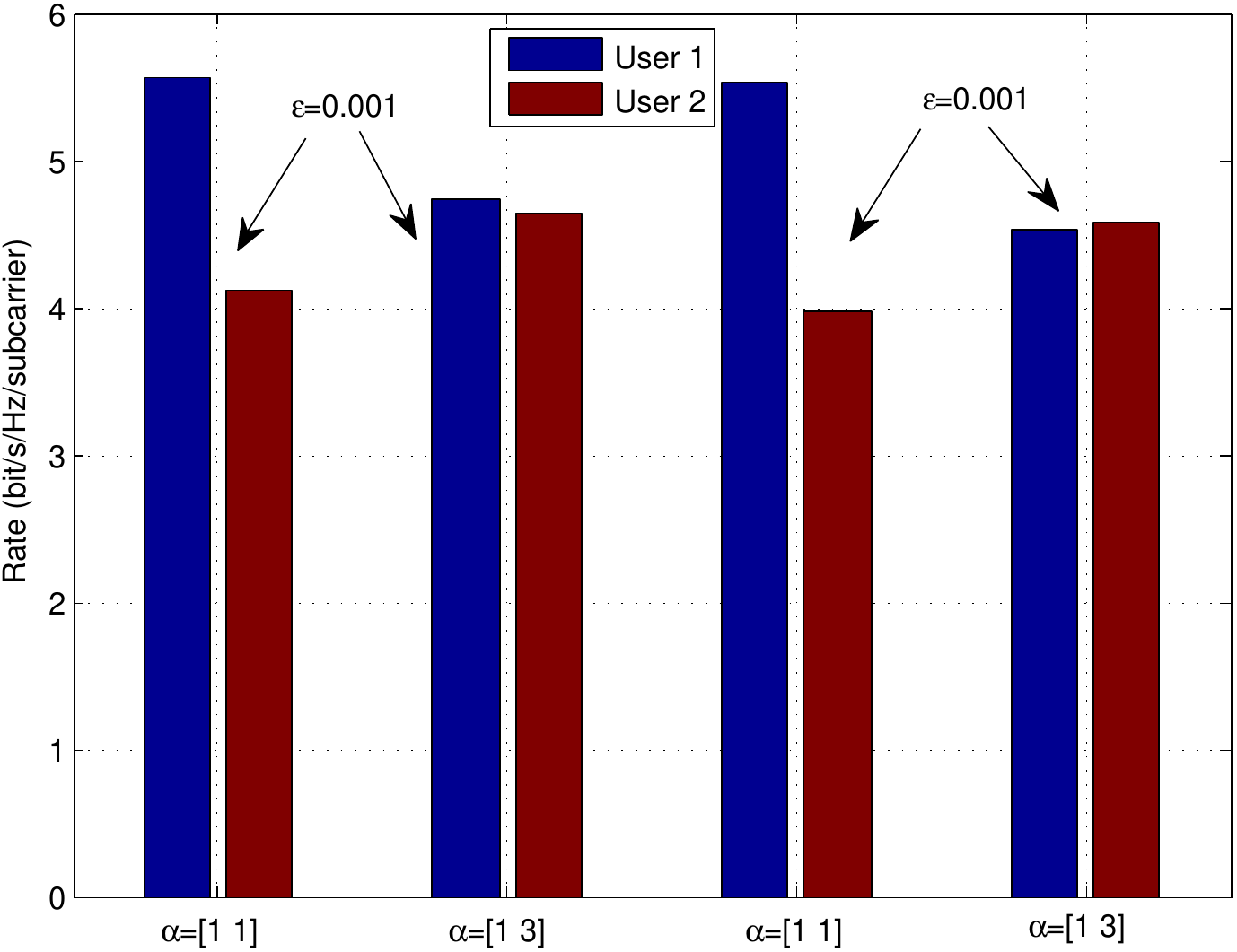}
		\caption{Rate comparison under different weights.}
		\label{figure8}
	\end{center}
\end{figure}

Figs.~\ref{figure1} and~\ref{figure2} show the convergence performance of the proposed inner iterative algorithm for solving the digital beamforming and reflection matrix, respectively, i.e., Line $5\sim 9$ and Line $11\sim 15$ in Algorithm~\ref{algorithm1}. Here, we set $\varepsilon=0$, $P_{\rm{max}}=0$ dBm and $\alpha_m=1\;(m\in\{1,2\})$. The legend ``$n$th iteration'' in Figs.~\ref{figure1} and~\ref{figure2} stands for the outer iteration number. One can observe that the inner iterative algorithm tends to converge after 5 iterations for each outer iteration. In addition, it can be found  that the gap  is small between the $2$nd and $3$rd iterations, but large between the $1$st and $2$nd iteration. This means that outer iterative loop (i.e., Algorithm~{\ref{algorithm1}}) also converges rapidly, and we will elaborate this details later in Fig.~\ref{figure3}. 

The convergence performance of Algorithm~\ref{algorithm1} under different estimation errors is plotted in Fig.~\ref{figure3}, where we set the maximum transmit power $P_{\rm{max}}=4$ dBm and $\alpha_m=1\;(m\in\{1,2\})$. It is clear that the rate tends to stabilize after 3 iterations, which demonstrates the fast  convergence of the proposed algorithm. In addition, it is easy to understand that the rate is low for a large estimation error as shown in Fig.~\ref{figure3}, where $\varepsilon=0$ means perfect CSIs between IRS and users.

Fig.~\ref{figure4} shows the rate versus $P_{\rm{max}}$ under different estimation errors, where we set $\alpha_m=1\;(m\in\{1,2\})$. Meanwhile, we also plot the rate under fully the digital structure, namely each antenna is connected to each RF chain. It is clear that the rate under the fully digital structure is higher than that under the sparse RF chain structure for the same condition, while the circuit power consumption is very high for the former. This is also one of the reasons for which the sparse RF chain structure is usually adopted when the ultra high frequency carrier is applied. In addition, one can observe that the rate increases with $P_{\rm{max}}$. 

We plot the rate under different BS antennas in Fig.~\ref{figure5}, where we set $P_{\rm{max}}=4$ dBm and $\alpha_m=1\;(m\in\{1,2\})$. It is obvious that the rate increases with the number of  antennas, but with a decreasing slope. In addition, we also plot the rate versus $P_{\rm{max}}$ for different numbers of IRS reflection elements in Fig.~\ref{figure6}, with $\varepsilon=0$. It is observed again that a large number of IRS reflection elements leads to a higher rate. This is because a higher beamforming gain can be achieved when there are either more antennas at the BS or more reflection elements at the IRS. 

To compare the system performance under different users' weights, we set the user weight $\alpha=[1\;\;1]$ and  $\alpha=[1\;\; 3]$ as shown in Figs.~\ref{figure7} and~\ref{figure8}. $\alpha=[1\;\;1]$ means that user~1 and user~2 have the same weight, whereas $\alpha=[1\;\;3]$ means that user~2 has a high priority than user~1. One can observe that the rate with $\alpha=[1\;\;3]$ is always lower  than that with $\alpha=[1\;\;1]$. Furthermore, we find that the rate of user~1 is higher than that of user~2 with $\alpha=[1\;\;1]$, while their rates are almost the same  for $\alpha=[1\;\;3]$. This means when $\alpha=[1\;\;1]$, all users have the same priority, and the system always optimizes the power or beamforming to maximize the sum rate. However, when user~2 has higher priority, the system allocates more power to user~2. In this case, although it can improve the rate of user~2, both the rate of user~1 and the sum rate are reduced. In practice, we can set different weights according to different quality of service requirements of the users. 

\section{Conclusion}
In this paper, we have considered an IRS-aided THz MIMO-OFDMA system, where the BS is equipped with a sparse RF chain structure. First, we have proposed a joint hybrid analog/digital beamforming and reflection matrix design to maximize the weighted sum rate under perfect CSIs. Next, considering the imperfect CSIs from the IRS to users, we have redesigned a robust joint optimization algorithm. From simulation results, we have found that the channel estimation error has a large impact on the system sum rate. In addition, one user with a high weight can improve this user's rate, while the system sum rate decreases.  Consequently, channel estimation schemes and users weight selection are important considerations for the design of practical systems.  
\appendices
\section{Proof of Theorem~\ref{theorem1}}
First, we give the Lagrangian function of (\ref{OptH}) without the rank-one constraint as
\begin{eqnarray}
\begin{aligned}
	&F(t_{m,k}, b_{m,k},\xi,\psi_{m,k},\nu_{m,k},{\bf{V}}_m[k])=\sum_{m=1}^{M}\sum_{k=1}^{K}a_m\log\left(1\!+\!t_{m,k}\right)\\
	+&\xi\left(P_{\rm{max}}-\sum\nolimits_{M=1}^M\sum\nolimits_{k=1}^K{\rm{Tr}}({\bf{F}}^H{\bf{F}}{\bf{V}}_m[k])\right)\\
	+&\sum_{m=1}^{M}\sum_{k=1}^{K}\psi_{m,k}\left(b_{m,k}-\left(\sum\nolimits_{j\neq m}^{M}{\rm{Tr}}({\bar{\bf{H}}}_m[k]{\bf{V}}_j[k])\!+\!\delta^2\right)\right)\\
	+&\sum_{m=1}^{M}\sum_{k=1}^{K}\nu_{m,k}\left( {\rm{Tr}}({\bar{\bf{H}}}_m[k]{\bf{V}}_m[k])-\left(\frac{t_{m,k}^{(n)}}{2b_{m,k}^{(n)}}b_{m,k}^2+\frac{b_{m,k}^{(n)}}{2t_{m,k}^{(n)}}t_{m,k}^2\right)\right)\\
	+&\sum_{m=1}^{M}\sum_{k=1}^{K}{\rm{Tr}}(\Theta_{m,k}{\bf{V}}_m[k]).
\end{aligned}
\end{eqnarray}
where $\xi,\psi_{m,k},\nu_{m,k}$, $\Theta_{m,k}$, respectively, represent the Lagrange multipliers corresponding the constraints (\ref{OptE2}), (\ref{eq23B}), (\ref{eq25}) and (\ref{OptH3}).  Since the relaxed SDP problem (\ref{OptH}) is convex, and the gap between the primal problem and its dual problem is zero, namely it satisfies the Slater's condition~\cite{convex}. Therefore, the Karaush-Kuhn-Tucker (KKT) conditions are necessary and sufficient for the optimal solutions of problem (\ref{OptH}) with rank-one constraint. Next, we give the KKT conditions related to the optimal digital beamforming ${\bf{V}}_m[k]^\star$ as:
\begin{subequations}
	\begin{eqnarray}
		\xi^\star {\bf{F}}^H{\bf{F}}+\sum_{j\neq m}^{M}\sum_{k=1}^{K}\psi_{m,k}^\star{\bar{\bf{H}}}_j[k]-\nu_{m,k}^\star{\bar{\bf{H}}}_m[k]&={\bf{\Theta}}_{m,k}^\star,\\
		{\bf{\Theta}}_{m,k}^\star{\bf{V}}_m[k]^\star&={\bf{0}},\label{App2}\\
		{\bf{\Theta}}_{m,k}^\star&\succeq {\bf{0}},
	\end{eqnarray}
\end{subequations}
where $\xi^\star$, $\psi_{m,k}^\star$, $\nu_{m,k}^\star$ and ${\bf{\Theta}}_{m,k}^\star$ are the optimal Lagrange multipliers.  The analog beamforming can be expressed as ${\bf{F}}=[{\bf{f}}_1,\cdots,{\bf{f}}_{\rm{RF}}]$, and we have ${\bf{f}}_i^H{\bf{f}}_i=1$ and ${\bf{f}}_i^H{\bf{f}}_j\ll 1\;(i\neq j)$  for a large $N_{\rm{TX}}$ as our analysis for the analog beamforming matrix in~Section~\ref{A1}.  In this way, for a large $N_{\rm{TX}}$, we can obtain ${\bf{F}}^H{\bf{F}}\approx {\bf{I}}$, and ${\bf{F}}^H{\bf{F}}$ is a full rank matrix, namely ${\rm{rank}}({\bf{F}}^H{\bf{F}})=N_{\rm{RF}}$. Because  $\xi^\star>0$, $\psi_{m,k}^\star>0$, we define 
\begin{eqnarray}
	{\bf{Y}}=\xi^\star {\bf{F}}^H{\bf{F}}+\sum_{j\neq m}^{M}\sum_{k=1}^{K}\psi_{m,k}^\star{\bar{\bf{H}}}_j[k],
\end{eqnarray}
and thus, ${\bf{Y}}$ is a positive-definite matrix which has full rank with ${\rm{rank}}({\bf{Y}})=N_{\rm{RF}}$. Based on this, we have 
\begin{eqnarray}
\begin{aligned}
  {\rm{rank}}({\bf{\Theta}}_{m,k}^\star)=&{\rm{rank}}({\bf{Y}}-\nu_{m,k}^\star{\bar{\bf{H}}}_m[k])\\
  \geq&{\rm{rank}}({\bf{Y}})-{\rm{rank}}(\nu_{m,k}^\star{\bar{\bf{h}}}_m[k]^H{\bar{\bf{h}}}_m[k])\\
  \geq&N_{\rm{RF}}-1.
  \end{aligned}
\end{eqnarray}

Therefore, we can claim that the rank of ${\bf{\Theta}}_{m,k}^\star$ is either $N_{\rm{RF}}$ or $N_{\rm{RF}}-1$. If ${\rm{rank}}({\bf{\Theta}}_{m,k}^\star)=N_{\rm{RF}}$, according to (\ref{App2}), the optimal ${\bf{V}}_m[k]^\star={\bf{0}}$, which means that the BS does not transmit any signal.  Thus, we have ${\rm{rank}}({\bf{\Theta}}_{m,k}^\star)=N_{\rm{RF}}-1$, and the null space of ${\bf{\Theta}}_{m,k}^\star$ is one dimensional. Meanwhile,  (\ref{App2}) means that ${\bf{V}}_m[k]^\star$ must lie in the null-space of ${\bf{\Theta}}_{m,k}^\star$, and we have ${\rm{rank}}({\bf{V}}_m[k]^\star)=1$ and the proof is completed. 

\bibliographystyle{IEEEtran}
\balance
\bibliography{references}

\begin{thebibliography}{10}
\providecommand{\url}[1]{#1}
\csname url@samestyle\endcsname
\providecommand{\newblock}{\relax}
\providecommand{\bibinfo}[2]{#2}
\providecommand{\BIBentrySTDinterwordspacing}{\spaceskip=0pt\relax}
\providecommand{\BIBentryALTinterwordstretchfactor}{4}
\providecommand{\BIBentryALTinterwordspacing}{\spaceskip=\fontdimen2\font plus
\BIBentryALTinterwordstretchfactor\fontdimen3\font minus
  \fontdimen4\font\relax}
\providecommand{\BIBforeignlanguage}[2]{{%
\expandafter\ifx\csname l@#1\endcsname\relax
\typeout{** WARNING: IEEEtran.bst: No hyphenation pattern has been}%
\typeout{** loaded for the language `#1'. Using the pattern for}%
\typeout{** the default language instead.}%
\else
\language=\csname l@#1\endcsname
\fi
#2}}
\providecommand{\BIBdecl}{\relax}
\BIBdecl

\bibitem{Cisco}
Cisco, ``Cisco annual internet report (20182023),'' \emph{White Paper}, 2020.

\bibitem{Sarieddeen_JSAC_2019}
H.~{Sarieddeen}, M.~{Alouini}, and T.~Y. {Al-Naffouri}, ``Terahertz-band
  ultra-massive spatial modulation {MIMO},'' \emph{IEEE J. Sel. Areas Commun.},
  vol.~37, no.~9, pp. 2040--2052, Sep. 2019.

\bibitem{Hao_IEEESJ_2020}
R.~{Zhang}, W.~{Hao}, G.~{Sun}, and S.~{Yang}, ``Hybrid precoding design for
  wideband {THz} massive {MIMO-OFDM} systems with beam squint,'' \emph{IEEE
  Sys. J.}, pp. 1--4, to be published, 2020.

\bibitem{Priebe_IEEETWC_2013}
S.~{Priebe} and T.~{Kurner}, ``Stochastic modeling of {THz} indoor radio
  channels,'' \emph{IEEE Trans. Wireless Commun.}, vol.~12, no.~9, pp.
  4445--4455, Sep. 2013.

\bibitem{Lu_2014_JSTSP}
L.~{Lu}, G.~Y. {Li}, A.~L. {Swindlehurst}, A.~{Ashikhmin}, and R.~{Zhang}, ``An
  overview of massive mimo: Benefits and challenges,'' \emph{IEEE J. Sel.
  Topics Signal Proc.}, vol.~8, no.~5, pp. 742--758, May 2014.

\bibitem{Zeng1}
M.~{Zeng}, X.~{Li}, G.~{Li}, W.~{Hao}, and O.~A. {Dobre}, ``Sum rate
  maximization for {IRS}-assisted uplink {NOMA},'' \emph{arXiv e-prints}, p.
  arXiv:2004.10791, 2020.

\bibitem{Zeng2}
M.~{Zeng}, E.~B. {Mohamed}, O.~A. {Dobre}, P.~{Fortier}, and Q.~V. {Pham},
  ``Energy-efficient resource allocation for {IRS}-assisted multi-antenna
  uplink systems,'' \emph{arXiv e-prints}, p. arXiv:2007.10002, 2020.

\bibitem{Wu_IEEETWC_2019}
Q.~{Wu} and R.~{Zhang}, ``Intelligent reflecting surface enhanced wireless
  network via joint active and passive beamforming,'' \emph{IEEE Trans.
  Wireless Commun.}, vol.~18, no.~11, pp. 5394--5409, Nov. 2019.

\bibitem{Chu_IEEEWCL_2020}
Z.~{Chu}, W.~{Hao}, P.~{Xiao}, and J.~{Shi}, ``Intelligent reflecting surface
  aided multi-antenna secure transmission,'' \emph{IEEE Wireless Commun.
  Lett.}, vol.~9, no.~1, pp. 108--112, Jan. 2020.

\bibitem{Gao_IEEEJSAC_2016}
X.~{Gao}, L.~{Dai}, S.~{Han}, C.~{I}, and R.~W. {Heath}, ``Energy-efficient
  hybrid analog and digital precoding for mmwave {MIMO} systems with large
  antenna arrays,'' \emph{IEEE J. Sel. Areas Commun.}, vol.~34, no.~4, pp.
  998--1009, Apr. 2016.

\bibitem{Lin_IEEETCOM_2015}
C.~{Lin} and G.~Y. {Li}, ``Adaptive beamforming with resource allocation for
  distance-aware multi-user indoor terahertz communications,'' \emph{IEEE
  Trans. Wireless Commun.}, vol.~63, no.~8, pp. 2985--2995, Aug. 2015.

\bibitem{Lin_IEEETWC_2016}
------, ``Energy-efficient design of indoor mmwave and sub-{THz} systems with
  antenna arrays,'' \emph{IEEE Trans. Wireless Commun.}, vol.~15, no.~7, pp.
  4660--4672, Jul. 2016.

\bibitem{Busari_IEEETVT_2019}
S.~A. {Busari}, K.~M.~S. {Huq}, S.~{Mumtaz}, J.~{Rodriguez}, Y.~{Fang}, D.~C.
  {Sicker}, S.~{Al-Rubaye}, and A.~{Tsourdos}, ``Generalized hybrid beamforming
  for vehicular connectivity using {THz} massive {MIMO},'' \emph{IEEE Trans.
  Veh. Tech.}, vol.~68, no.~9, pp. 8372--8383, Sep. 2019.

\bibitem{Tan_IEEEGB_2019}
J.~{Tan} and L.~{Dai}, ``Delay-phase precoding for {THz} massive {MIMO} with
  beam split,'' in \emph{Proc. IEEE GLOBECOM}, 2019, pp. 1--6.

\bibitem{Hang_IEEETCOM_2019}
H.~{Yuan}, N.~{Yang}, K.~{Yang}, C.~{Han}, and J.~{An}, ``Hybrid beamforming
  for terahertz {MIMO-OFDM} systems over frequency selective fading,''
  \emph{arXiv e-prints}, p. arXiv:1910.05967, 2019.

\bibitem{Ning_2019arXiv}
B.~{Ning}, Z.~{Chen}, W.~{Chen}, Y.~{Du}, and J.~{Fang}, ``Terahertz multi-user
  massive mimo with intelligent reflecting surface: Beam training and hybrid
  beamforming,'' \emph{arXiv e-prints}, p. arXiv:1912.11662, 2019.

\bibitem{Zhang_IEEEJSAC_2020}
S.~{Zhang} and R.~{Zhang}, ``Capacity characterization for intelligent
  reflecting surface aided {MIMO} communication,'' \emph{IEEE J. Sel. Areas
  Commun.}, to be published, 2020.

\bibitem{Yang_IEEETCOM_2020}
Y.~{Yang}, B.~{Zheng}, S.~{Zhang}, and R.~{Zhang}, ``Intelligent reflecting
  surface meets {OFDM}: Protocol design and rate maximization,'' \emph{IEEE
  Trans. Commun.}, vol.~68, no.~7, pp. 4522--4535, Jul. 2020.

\bibitem{Dai_arXiv_2020}
Z.~{Zhang} and L.~{Dai}, ``A joint precoding framework for wideband
  reconfigurable intelligent surface-aided cell-free network,'' \emph{arXiv
  e-prints}, p. arXiv:2002.03744, 2020.

\bibitem{Zhou_IEEEWCL_2020}
G.~{Zhou}, C.~{Pan}, H.~{Ren}, K.~{Wang}, M.~{Di Renzo}, and A.~{Nallanathan},
  ``Robust beamforming design for intelligent reflecting surface aided {MISO}
  communication systems,'' \emph{IEEE Wireless Commun. Lett.}, to be published,
  2020.

\bibitem{Priebe_JCN_2013}
S.~{Priebe}, M.~{Kannicht}, M.~{Jacob}, and T.~{Kürner}, ``Ultra broadband
  indoor channel measurements and calibrated ray tracing propagation modeling
  at {THz} frequencies,'' \emph{J. Commun. and Net.}, vol.~15, no.~6, pp.
  547--558, Jun. 2013.

\bibitem{Cui_2019arXiv}
W.~{Tang}, M.~Z. {Chen}, X.~{Chen}, J.~Y. {Dai}, Y.~{Han}, M.~{Di Renzo},
  Y.~{Zeng}, S.~{Jin}, Q.~{Cheng}, and T.~J. {Cui}, ``Wireless communications
  with reconfigurable intelligent surface: Path loss modeling and experimental
  measurement,'' \emph{arXiv e-prints}, p. arXiv:1911.05326, 2019.

\bibitem{Park_IEEETWC_2017}
S.~{Park}, A.~{Alkhateeb}, and R.~W. {Heath}, ``Dynamic subarrays for hybrid
  precoding in wideband mmwave {MIMO} systems,'' \emph{IEEE Trans. Wireless
  Commun.}, vol.~16, no.~5, pp. 2907--2920, May 2017.

\bibitem{shen_IEEETSP_2019}
W.~{Shen}, X.~{Bu}, X.~{Gao}, C.~{Xing}, and L.~{Hanzo}, ``Beamspace precoding
  and beam selection for wideband millimeter-wave {MIMO} relying on lens
  antenna arrays,'' \emph{IEEE Trans. Signal Process.}, vol.~67, no.~24, pp.
  6301--6313, Dec. 2019.

\bibitem{Song_IEEEICASSP_2016}
P.~{Song}, G.~{Scutari}, F.~{Facchinei}, and L.~{Lampariello}, ``D3m:
  Distributed multi-cell multigroup multicasting,'' in \emph{Proc. IEEE
  International Conference on Acoustics, Speech and Signal Processing
  (ICASSP)}, 2016, pp. 3741--3745.

\bibitem{Qi_IEEETVT_2016}
Q.~{Zhang}, Q.~{Li}, and J.~{Qin}, ``Robust beamforming for nonorthogonal
  multiple-access systems in {MISO} channels,'' \emph{IEEE Trans. Veh.
  Technol.}, vol.~65, no.~12, pp. 10\,231--10\,236, Dec. 2016.

\bibitem{Pan_IEEETWC_2018}
Z.~{Yang}, C.~{Pan}, W.~{Xu}, Y.~{Pan}, M.~{Chen}, and M.~{Elkashlan}, ``Power
  control for multi-cell networks with non-orthogonal multiple access,''
  \emph{IEEE Trans. Wireless Commun.}, vol.~17, no.~2, pp. 927--942, Feb. 2018.

\bibitem{Luo_IEEESPM_2010}
Z.~{Luo}, W.~{Ma}, A.~M. {So}, Y.~{Ye}, and S.~{Zhang}, ``Semidefinite
  relaxation of quadratic optimization problems,'' \emph{IEEE Signal Process.
  Mag.}, vol.~27, no.~3, pp. 20--34, Mar. 2010.

\bibitem{Chi_IEEETSP_2014}
K.~{Wang}, A.~M. {So}, T.~{Chang}, W.~{Ma}, and C.~{Chi}, ``Outage constrained
  robust transmit optimization for multiuser {MISO} downlinks: Tractable
  approximations by conic optimization,'' \emph{IEEE Trans. Signal Process.},
  vol.~62, no.~21, pp. 5690--5705, Nov. 2014.

\bibitem{Zhou_IEEETWC_2017}
F.~{Zhou}, Z.~{Li}, J.~{Cheng}, Q.~{Li}, and J.~{Si}, ``Robust {AN}-aided
  beamforming and power splitting design for secure {MISO} cognitive radio with
  {SWIPT},'' \emph{IEEE Trans. Wireless Commun.}, vol.~16, no.~4, pp.
  2450--2464, Apr. 2017.

\bibitem{convex}
S.~{Boyd} and L.~{Vandenberghe}, ``Convex optimization,'' \emph{Cambridge
  University Press}, 2004.

\end{thebibliography}

\end{document}